%% file: main.tex
\documentclass[11pt]{article}
\input{newcommands}

\title{A Polymatroidal Perspective on Random Contraction
\thanks{Grainger College of Engineering. Univ. of Illinois, Urbana-Champaign, Urbana, IL 61801. Email: {\tt
      \{karthe, chekuri, weihaoz3\}@illinois.edu}. Supported in part by NSF grant CCF-2402667.}
}

\author{Karthekeyan Chandrasekaran
\and Chandra Chekuri
\and Weihao Zhu}

\date{}
\begin{document}

\maketitle

\pagenumbering{gobble}
\input{abstract}


\newpage
\pagenumbering{arabic}

\input{intro-3}

\input{preliminary}
\input{nonuniform-random-contraction-3}

\input{approx-random-contraction}
\input{parallel-random-contraction}
\input{strong-uniform-random-contraction-bounded-marginal}
\input{quotient-bounded-instance}
\input{conclusion}

\paragraph{AI Disclosure.} The authors used ChatGPT to assist with the writing of this manuscript. The authors assume full responsibility for all content. 

\bibliographystyle{abbrv}
\bibliography{references}

\newpage
\appendix
\input{quotient-property}
\input{gkp-counting-bound}
\input{uniform-random-contraction-bounded-marginal}
\input{lp-analysis}

\end{document}

%% file: newcommands.tex
\usepackage{complexity}

\usepackage[english]{babel}
\usepackage[utf8x]{inputenc}
\usepackage[T1]{fontenc}
\usepackage{enumitem}
\usepackage[margin=1.in,marginparwidth=1.75cm]{geometry}

\usepackage[algo2e, ruled, vlined]{algorithm2e}
\usepackage{amsmath}
\usepackage{amsfonts}
\usepackage{graphicx}
\usepackage{amsthm}
\usepackage{amssymb}
\usepackage{dsfont}
\usepackage[algo2e, ruled, vlined]{algorithm2e}
\usepackage[noend]{algorithmic} 
\usepackage{bbm}

\usepackage[T1]{fontenc}
\usepackage{microtype}          %
\usepackage[full]{textcomp}

\usepackage{lipsum}
\usepackage{xspace}
\usepackage{relsize}
\usepackage{framed}
\usepackage{xcolor}
\usepackage{graphicx}
\usepackage{multirow}
\usepackage{subcaption}
\usepackage{float}

\usepackage[colorinlistoftodos]{todonotes}
\usepackage[colorlinks=true, allcolors=blue]{hyperref}

\newtheorem{theorem}{Theorem}

\newtheorem{lemma}{Lemma}
\newtheorem{corollary}{Corollary}

\newtheorem{remark}{Remark}

\theoremstyle{definition}
\newtheorem{definition}{Definition}

\newcommand{\calN}{{\mathcal{N}}}

\newcommand{\znote}[1]{{\bf{\color{red}[\tiny Weihao: #1]}}}
\newcommand{\knote}[1]{{\bf{\color{blue}[\tiny Karthik: #1]}}}
\newcommand{\cnote}[1]{{\bf{\color{magenta}[\tiny Chandra: #1]}}}

\newcommand{\mypara}[1]{\medskip \noindent {\bf #1}}

\usepackage{float}
\usepackage{booktabs}
\usepackage{multirow}
\usepackage{thmtools} 
\usepackage{thm-restate}
\usepackage{lipsum}

\newcommand{\Z}{\mathbb{Z}}

\newcommand{\cF}{\mathcal{F}}

\newcommand{\cN}{\mathcal{N}}
\newcommand{\cI}{\mathcal{I}}
\newcommand{\cQ}{\mathcal{Q}}
\newcommand{\cM}{\mathcal{M}}

\newif\ifdraft
\draftfalse

\newcommand{\MinQuo}{\textsc{Min-Quotient}\xspace}

\newcommand{\GraphMinCut}{\textsc{Graph-Min-Cut}\xspace}
\newcommand{\GraphkCut}{\textsc{Graph-$k$-Cut}\xspace}
\newcommand{\HypergraphMinCut}{\textsc{Hypergraph-Min-Cut}\xspace}
\newcommand{\HypergraphkCut}{\textsc{Hypergraph-$k$-Cut}\xspace}
\newcommand{\HedgeMinCut}{\textsc{Hedgegraph-Min-Cut}\xspace}
\newcommand{\Comp}{\#\text{Comps}}
\newcommand{\Span}{\text{span}}
\newcommand{\qbounded}{\text{quotient-bounded}\xspace}
\newcommand{\sqbounded}{\text{strongly-quotient-bounded}\xspace}

%% file: abstract.tex
\begin{abstract}\label{section:abstract}
  Karger's elegant random contraction algorithm for finding a global mincut in a graph \cite{karger1993global,karger1995random} has been highly influential. More recent work has obtained several different (nonuniform) random contraction algorithms for mincut in hypergraphs and hedgegraphs \cite{kogan2015sketching, GKP17, CXY21, fox2023minimum, fomin2025fixed}.  Motivated by the conceptual goal of understanding these algorithms in a unified fashion, we study random contraction algorithms for finding a minimum quotient of a polymatroid. We introduce the notion of \qbounded polymatroids and show that several existing results can be derived and understood under a common algorithmic framework for \qbounded polymatroids.   
\end{abstract}

%% file: intro-3.tex
\section{Introduction}\label{section:introduction}
In \GraphMinCut problem the input contains a connected multigraph $G=(V,E)$ with $n:=|V|$ vertices and the goal is to find a minimum number of edges whose removal disconnects the vertex set. Karger developed a simple and highly influential random contraction algorithm for this problem \cite{karger1993global}. It is easy to describe it in a recursive fashion. In the base case, when $|V|=2$, the algorithm returns the unique non-trivial cut in the graph between the two vertices. Otherwise it picks an edge $e \in E$ \emph{uniformly} at random and contracts it to obtain a graph $G' = G/e$, and recurses on $G'$. Karger proved that for every specific mincut $C\subseteq E$, this algorithm returns $C$ with probability at least $\frac{1}{{n \choose 2}}$. The analysis is also simple, and is routinely taught as a first randomized algorithm.  There have been many important and influential developments related to this original work, several of which were already in Karger's thesis \cite{karger1995random}.  
In this paper we focus mainly on polynomial-time solvability and structural aspects, and hence, we will not dwell on several refinements of the random contraction algorithms that yield faster algorithms such as \cite{KS96}.

More recently, random contraction has been studied for computing the minimum cut in hypergraphs, and a generalization of hypergraphs called hedgegraphs. A hypergraph $H=(V,E)$ consists of a finite vertex set $V$ and a set $E$ of hyperedges, where each hyperedge $e\in E$ is a subset of vertices. A hedgegraph $G=(V,E)$ consists of a finite vertex set $V$ and a finite set $E$ of hedges, where each hedge $e\in E$ is a collection of distinct hyperedges. Hedgegraphs were introduced by Ghaffari, Karger, and Panigrahi \cite{GKP17} to model correlated edge failures in graphs. The global mincut problem in these two structures is easy to define. Find a partition of $V$ into $(S, V-S)$ to minimize the number of hyperedges (hedges, respectively) that cross the partition --- equivalently, find a minimum number of hyperedges (hedges, respectively) to remove in order to disconnect the hypergraph (or hedgegraph). Ghaffari, Karger and Panigrahi \cite{GKP17} showed that a certain \emph{nonuniform} random contraction algorithm yields a polynomial-time algorithm for \HypergraphMinCut problem. Interestingly they also developed a quasi-polynomial-time algorithm for \HedgeMinCut and a PTAS for it, again via a nonuniform version of random contraction. Subsequent work showed that \HedgeMinCut is hard \cite{JLMPTS26} --- under the exponential-time hypothesis (ETH) there is no polynomial-time algorithm for it! Chandrasekaran, Xu, and Yu \cite{CXY21} developed a different variant of nonuniform random contraction that also yields a polynomial-time algorithm for  \HypergraphMinCut, and also for the special case of \HedgeMinCut in constant span hedgegraphs (we will formally define span later). Fomin, Golovach, Korhonen, Lokshtanov, and Saurabh \cite{fomin2025fixed} designed a different variant of nonuniform random contraction that yields a fixed-parameter algorithm for \HedgeMinCut parameterized by solution size. 

A natural question is whether the aforementioned results can be understood in a common framework. We describe a a second motivation for our work. Karger, subsequent to his work on graphs, showed that random contraction is natural from a \emph{matroidal} perspective and that the algorithm to compute a global mincut can be viewed as computing the minimum \emph{quotient} in the underlying graphic matroid. One could then ask if random contraction works for every matroid to find a minimum quotient. However, finding a minimum quotient is NP-Hard even for restricted classes of matroids \footnote{Finding a minimum quotient in the dual of transversal matroids is NP-hard---attributed to Stockmeyer in \cite{mccormick-thesis}.}. Despite this, the underlying random sampling machinery can be extended to matroids --- Karger pointed this out \cite{Kar-matroid} and derived several applications. More recently, Quanrud \cite{quanrud2024quotient}, building upon some of Karger's work, showed that \emph{every} matroid (and polymatroid) admits quotient sparsification. This generalized well-known results on graph cut sparsification by Benczur and Karger \cite{BK15}, and hypergraph cut sparsification by Chen, Khanna and Nagda \cite{chen2020near}, and led to new results. In another direction, Chandrasekaran, Chekuri, and Zhu in recent work \cite{chandrasekaran2025hedgegraph} showed that it is fruitful to view hedgegraphs from a polymatroidal view even though the cut function (that is defined over the vertex set) is not submodular. These motivate us to understand the successes of random contraction for
graphs, hypergraphs, and hedgegraphs from a polymatroidal perspective, and develop a unified picture that is not  only explanatory but also has the scope to generalize and yield new results and insights. We first describe relevant background.

\mypara{Polymatroids, Matroids, and Quotients:}
A \emph{polymatroid} $f:2^{\cN}\rightarrow \mathbb{Z}_{\geq 0}$ on a ground set $\cN$ is an integer-valued monotone submodular function with $f(\emptyset)=0$. A function $f: 2^{\cN}\rightarrow \mathbb{Z}_{\geq 0}$ is \emph{monotone} if $f(A)\leq f(B)$ for every $A\subseteq B$ and \emph{submodular} if $f(A)+f(B)\geq f(A\cup B)+f(A\cap B)$ for every $A,B\subseteq \cN$. Equivalently, $f$ is submodular if $f(A + e) - f(A) \ge f(B+e) - f(B)$ for all $A \subset B$ and $e \in \cN \setminus B$. 
We define the \emph{rank} of $f$ to be $f(\cN)$. 
The \emph{span} of a set $A$, denoted $\Span(A):=\{e \in \cN \mid f(A \cup\{e\}) \le f(A) \}$.  A subset $S\subseteq \cN$ is \emph{closed} if $f(S\cup\{e\})>f(S)$ for every $e\in \cN\setminus S$, in other words if $\Span(S) = S$. We call a subset $Q\subseteq\cN$ a \emph{quotient} if $\cN\setminus Q$ is a closed set. We note that $\cN$ is always a closed set, implying that the empty set $\emptyset$ is a quotient. A matroid $\cM = (\cN, \cI)$ consists of a finite ground set $\cN$ and a collection of \emph{independent} sets $\cI \subseteq 2^{\cN}$ that satisfy three properties. Since these are well-known, we instead define matroids as a special case of polymatroids. In this view, a matroid is defined by its rank function $r_{\cM}: 2^{\cN} \rightarrow \mathbb{Z}_{\geq 0}$ which is a polymatroid with an additional property: $r_{\cM}(e) \in \{0,1\}$ for each $e \in \cN$.  A set $A \subseteq \cN$ is independent iff $r_{\cM}(A) = |A|$. A \emph{circuit} is a minimal dependent set. A \emph{base} of a matroid is a maximal independent set. It is easy to establish that all bases in a matroid have the same cardinality, and this is called the rank of the matroid $\cM$. Note that in a polymatroid one can define a base to be a minimal set $S$ such that $f(S) = f(\cN)$, and in this setting, bases may have different cardinalities. We refer to \cite{Schrijver-book} for additional background. The main problem of interest to us is the following.

\begin{definition}[\MinQuo Problem]
  Given a polymatroid $f:2^{\cN}\rightarrow \mathbb{Z}_{\geq 0}$ via an
  evaluation oracle, return a minimum cardinality non-empty quotient of $f$.
\end{definition}

When $f$ is the rank function of a matroid $\cM$, then the min-quotient problem is the same as the well-known min-cocircuit or the cogirth problem. A co-circuit is a set that intersects every base of the matroid, or equivalently, it is a circuit in the dual matroid $\cM^*$. Finding a min-circuit of a matroid is an important algorithmic problem; for instance the shortest code-word in a binary linear code is a special case and Vardy showed that this is NP-Hard \cite{Var97}. 

\mypara{Quotients and cuts in graphs, hypergraphs, and hedgegraphs:} Let $G=(V,E)$ be a multigraph. The well-known graphic matroid associated with $G$ is $\cM_G=(E, \cI)$ where a set $A \subseteq E$ of edges is independent iff the subgraph $(V,A)$ is a forest. The rank function of this matroid, $r_{G}:2^E\rightarrow \mathbb{Z}_{\geq 0}$ is defined as: $r_G(A):=|V|-\Comp(V,A)$ for $A\subseteq E$, where $\Comp(V,A)$ is the number of connected components in the subgraph $(V,A)$. When $G$ is connected, non-empty quotients of $r_G$ correspond exactly to edge cutsets (a set of edges whose removal disconnects the graph): a quotient $Q$ is the set of edges that cross the connected components of the closed set $E \setminus Q$.  Thus, finding a minimum non-empty quotient of $r_G$ is precisely the \GraphMinCut problem\footnote{The connection between mincut and min-quotients is most natural when each element has unit weight. In the graph setting we can treat weighted edges by making parallel copies and the equivalence holds.}. 

Now we consider a hypergraph $H=(V,E)$.  Here, we consider the function $f:2^E \rightarrow \mathbb{Z}_{\ge 0}$ defined by $f(A) = |V| - \Comp(V,A)$, where $\Comp(V,A)$ is the number of connected components in the subgraph $(V,A)$. $f$ is a polymatroid and we denote it as the hypergraphic polymatroid associated with $H$ \cite{chandrasekaran2025hedgegraph}. For a hyperedge $e$,  $f(e) = |e|-1$ and thus in general $f$ is not a matroid. As in graphs, a quotient $Q \subseteq E$ corresponds to a hyperedge cutset (i.e., a set of hyperedges whose removal disconnects the hypergraph). Thus, finding a min-quotient of $f$ is the same as solving \HypergraphMinCut. 

Next we consider a hedgegraph $G=(V,E)$. We recall that each hedge $e\in E$ is a collection of distinct hyperedges. Consider the function $f:2^E \rightarrow \mathbb{Z}_{\ge 0}$ defined by $f(A) = |V| - \Comp(V,A)$, where $\Comp(V,A)$ is the number of connected components in the hypergraph $(V,\{h: h\in e\text{ for some }e\in A\})$. $f$ is a polymatroid and we denote it as the hedgegraph polymatroid associated with $G$ \cite{chandrasekaran2025hedgegraph}.  Finding a min-quotient of $f$ is the same as solving \HedgeMinCut. 

We point out an important aspect in viewing the mincut problems from the min-quotient view. 
In both graphs and hypergraphs, the cut function $d_G:2^V \rightarrow \mathbb{Z}_{\ge 0}$ defined over the \emph{vertices} is a nonmonotone submodular set function; thus, one can solve \GraphMinCut and \HypergraphMinCut via submodular set function minimization --- in fact, one can solve the min-$s$-$t$-cut problem as well. In contrast, the cut function of a hedgegraph is \emph{not} submodular \cite{GKP17}, and min-$s$-$t$-cut problem in hedgegraphs is NP-Hard \cite{ZCTZ11}! 

\mypara{$c$-\qbounded Polymatroid Families:} The min-cocircuit problem in matroids and polymatroids is NP-Hard.
We define a parameterized family of minor-closed polymatroids and show that random contraction algorithms succeed in solving the min-quotient problem for them. 

\begin{definition}[Minor-Closed Polymatroid Family]
  Let $f:2^{\cN}\rightarrow \mathbb{Z}_{+}$ be a polymatroid and $e\in \cN$. The set function $f_{\backslash e}:2^{\cN -e}\rightarrow \mathbb{Z}_{+}$ is the polymatroid obtained by \emph{deleting} $e$, i.e., $f_{\backslash e}(S):=f(S)$ for all $S\subseteq \cN-e$. The polymatroid $f_e:2^{\cN -e}\rightarrow \mathbb{Z}_{+}$ is the polymatroid obtained by \emph{contracting} $e$, i.e., $f_e(S):=f(S+e)-f(e)$ for all $S\subseteq \cN-e$.  For a subset $\cN_1\subseteq \cN$, a polymatroid $g:2^{\cN_1}\rightarrow \mathbb{Z}_{\geq 0}$ is a \emph{minor} of $f$ if it can be obtained from $f$ by deleting elements and contracting elements. A family $\mathcal{F}$ of polymatroids is \emph{minor-closed} if for every $f \in \mathcal{F}$ all minors of $f$ are also in $\mathcal{F}$.
\end{definition}


\begin{definition}[Quotient Bounded Polymatroid Family]\label{definition:bicriteria-ratio}
  Let $\cF$ be a minor-closed family of polymatroids and $c \ge 1$ be an integer. The family $\cF$ is a \emph{$c$-\qbounded} if for every $f:2^{\cN}\rightarrow \mathbb{Z}_{\geq 0}$ in $\cF$, there is a non-empty quotient $Q$ of $f$ such that\footnote{For simplicity, we assume that the empty polymatroid always satisfies this property.}
  $$|Q| \leq c \cdot \frac{\sum_{e\in \cN}f(e)}{f(\cN)}.$$
  We say that $\cF$ is \emph{$c$-\sqbounded} if for every polymatroid $f:2^{\cN}\rightarrow \mathbb{Z}_{\geq 0}$ in $\cF$, there is a non-empty quotient $Q$ of $f$ such that
  $$|Q| \leq \frac{\sum_{e\in \cN}(f(e)+ c)}{f(\cN)}.$$
\end{definition}

If $\cF$ is the set of rank functions of a minor-closed family of matroids, then the two definitions are roughly equivalent---in particular, $c\cdot \sum_{e\in \cN}f(e)/f(\cN)=c\cdot |\cN|/f(\cN)$ and $\sum_{e\in \cN}(f(e) + c)=(c+1)\cdot|\cN|/f(\cN)$, i.e., they differ by a factor of at most $2$. In this setting, our definition is the same as the notion of $c$-ratio matroids defined by Karger \cite{Kar-matroid}: these are matroids with a gap of at most $c$ between the min-quotient and the base packing number. 
A \emph{$k$-polymatroid} is a polymatroid in which the function value of every singleton set is at most $k$. In this sense, $1$-polymatroids are equivalent to matroid rank functions. If $k$ is a fixed-constant, then $k$-polymatroids are close to being matroid rank functions; also $c$-\sqbounded  and $c$-\qbounded are equivalent upto a factor of $2k$: $c\cdot \sum_{e\in \cN}f(e)/f(\cN)\in[c |\cN|/f(\cN), ck|\cN|/f(\cN)]$ and $\sum_{e\in \cN}(f(e) + c)\in [(c+1)|\cN|/f(\cN),(c+k)\cdot|\cN|/f(\cN)]$. An important example that is relevant to this work are hypergraphs where each hyperedge has at most $k$ vertices; we call them $k$-hypergraphs\footnote{Typically these are called rank-$k$ hypergraphs but we avoid this due to the use of rank for other concepts.}. The hypergraphic polymatroid family associated with $k$-hypergraphs are in fact $(k-1)$-polymatroids. 
Our definitions bring out subtle differences in the setting of general polymatroids. For a minor-closed family of polymatroids, $c$-\sqbounded is a much stronger property than $c$-\qbounded, especially from the perspective of random contraction. We explain this next by considering three concrete polymatroid families that are minor-closed.

In graphs, the analysis of the random contraction algorithm of Karger relies on the fact that the global mincut value is at most the \emph{average degree} which is $2|E|/|V|$ --- this is true because there is a vertex $v$ with degree at most the average degree and the edges incident to $v$ form a valid cut.  A similar fact is also used in the analysis of the nouniform random contraction of hypergraphs: The global mincut value is at most the average degree which is $(\sum_{e \in E} |e|)/|V|$. Via these observations, one can show that both the hypergraphic polymatroid family and the graphic matroid rank function family are $2$-\qbounded. 
In fact, one can also show that
the hedgegraph polymatroid family is also $2$-\qbounded. What distinguishes
hypergraphs and hedgegraphs? A key observation is that the hypergraphic polymatroid family is \emph{$1$-\sqbounded}. 
The \emph{span} of a hedgegraph is the maximum number of vertex-disjoint components among all hedges. 
We also observe that the hedgegraph polymatroid family of hedgegraphs of span at most $s$ is $s$-\sqbounded. Section \ref{section:quotient-bounded-instance} proves these observations.

\mypara{On random contraction for polymatroids:} Random contraction algorithms use a series of contractions to reduce the problem to an easy base case where one can enumerate all the quotients/mincuts by brute force. The base case in the setting of graphs, hypergraphs, and hedgegraphs is when the number of vertices is a fixed constant. In the setting of polymatroids, this is when the rank is small. The algorithm needs to choose an element to contract to ensure that a min-quotient survives. In the setting of matroids, all elements $e$ have $f(e) = 1$ and are indistinguishable for a  random contraction algorithm that is oblivious to global structure, and hence, uniform random contraction is the natural choice. However, in the setting of polymatroids, $f(e)$ values can differ substantially, so contracting an element with large $f(e)$ value reduces the rank significantly in one step but may destroy the minimum quotient. For instance, consider the extreme case of an element $e$ with $f(e) = f(N)$; every such element belongs to every non-empty quotient. We cannot contract such an element; we need to include it in every quotient and delete it from further consideration. This discussion shows that a random contraction algorithm may need to differentiate the elements according to $f(e)$ value while being somewhat oblivious to the structure, and this naturally leads to various non-uniform random contraction approaches.
We note that the abstract polymatroid framework does not ``see'' the vertices when applied to the polymatroid associated with graphs, hypergraphs, and hedgegraphs; recall that the polymatroid is defined over edges/hyperedges/hedges. The $c$-\sqbounded and $c$-\qbounded conditions abstract the desired average degree bounds on the min-quotient in a clean fashion.

\subsection{Our Results}

Our conceptual contribution is bringing a polymatroidal lens to random contraction, and identifying the family of \qbounded polymatroids on which variants of random contraction succeed in finding an (approximate) min-quotient.
We state the informal version of our results below. 
In the following, $r$, $n$, and $q$ denote the rank, size of the ground set, and size of a minimum non-empty quotient of the input polymatroid respectively. For ease of understanding, we assume $c=O(1)$ in the following results. 
For a polymatroid $f$ and a value $\alpha>0$, a non-empty quotient $Q'$ is a \emph{$\alpha$-approximate minimum quotient} of $f$ if $|Q'|\le \alpha|Q|$ for every non-empty quotient $Q$ of $f$. We recall that a $k$-polymatroid is a polymatroid whose function value on every singleton set is at most $k$. 
\begin{enumerate}
    \item For $c$-\sqbounded polymatroid families, we design a random contraction algorithm to find a minimum non-empty quotient in time $(r\cdot n)^{O(c)}$, i.e., in polynomial time. 
    This recovers the randomized polynomial-time algorithm of \cite{CXY21} for \HedgeMinCut in bounded span hedgegraphs, hypergraphs, and graphs. 
    \item For $c$-\qbounded polymatroid families, we design a random contraction algorithm to find a $(1+\varepsilon)$-approximate minimum non-empty quotient in time $r^{O(c\log{(1/\varepsilon)})}n^{O(c)}$, i.e., a PTAS;
    this also implies an algorithm to find a minimum non-empty quotient in time $r^{O(c\log{q})}n^{O(c)}$, i.e., in quasi-polynomial time. 
     This recovers the PTAS and quasi-polynomial-time algorithms of \cite{GKP17} for \HedgeMinCut. 
    \item For $c$-\qbounded polymatroid families, we design a random contraction algorithm to find a minimum non-empty quotient in time $2^q\cdot r^{O(c\log c)} \cdot n^{O(c)}$, i.e., an FPT algorithm when parameterized by the solution size $q$. 
    This recovers the FPT algorithm of \cite{fomin2025fixed} for \HedgeMinCut parameterized by solution size. 
\end{enumerate}

The algorithms underlying our results also lead to counting bounds shown in Table \ref{table:number} below. 
\begin{table}[H]
\centering
\begin{tabular}{|l|l|l|}
\hline
\textbf{Family} & \textbf{Number of minimum Quotients}\\
\hline
$c$-\sqbounded & $r^{O(c)}\cdot n^{O(c)}$ (Corollary \ref{coro:sqbounded-number}) \\
$c$-\qbounded & $r^{O(c\log{q})}\cdot n^{O(c)}$ (Corollary \ref{coro:qbounded-number-1})\\
$c$-\qbounded & $2^q\cdot r^{O(c\log c)} \cdot n^{O(c)}$ (Corollary \ref{coro:qbounded-number-2})\\
\hline
\end{tabular}
\caption{Bounds on the number of minimum quotients in $c$-\qbounded and $c$-\sqbounded polymatroid families. 
The two bounds for $c$-\qbounded family are incomparable since one is smaller than the other in different regimes of $q$, $r$, and $c$.}
\label{table:number}
\end{table}

\paragraph{$k$-polymatroids.} 
A key consequence of Karger's random contraction is a bound of $O(n^{2\alpha})$ on the number of approximate min-cuts in $n$-vertex connected graphs. This bound on the number of approximate min-cuts has found applications in cut-sparsifiers, graph unreliability, and network design \cite{BK15, Kar01, Kar99-sampling}. 
The graph upper bound does not generalize to hypergraphs (see \cite{CXY21} for an example), but generalizes to connected $k$-hypergraphs for fixed constants $k$: a naive analysis of random contraction leads to a bound of $O(n^{k\alpha})$; Kogan and Krauthgamer improved this bound to $O(2^{k\alpha}kn^{2\alpha})$ \cite{kogan2015sketching}. 
Although the bound on approximate min-cuts in graphs have been obtained via other methods \cite{NNI97, CQX20, BCW23}, the only known proof for the bound for $k$-hypergraphs is via the random contraction algorithm. Our next result recovers this bound for $k$-hypergraphs through the polymatroidal lens. 
\begin{enumerate}[start=4]
    \item For $c$-\sqbounded $k$-polymatroid families, we design a random contraction algorithm to find a subset $\cN'$ of the ground set such that every fixed $\alpha$-approximate minimum quotient $Q$ of $f$ is contained in $\cN'$ with probability $\Omega(k^{-1} r^{-(c+1)\alpha})$ and moreover, the polymatroid obtained by contracting the complement of $\cN'$ has rank at most $\alpha(c+k)$.
\end{enumerate}

We recall that the family of hypergraphic matroids associated with $k$-hypergraphs is a $1$-\sqbounded $(k-1)$-polymatroid family. 
Moreover, the ground set of the hypergraphic matroid associated with a hypergraph is the set hyperedges $E$ and its rank is at most the number of vertices. Hence, the above result implies that the number of quotients (i.e., cut-sets) in the set $\cN'\subseteq E$ returned by the algorithm is $|\cN'|^{\alpha(c+k)}\le |E|^{\alpha k}$ (see Lemma \ref{lemma:quotient-counting}) and consequently, the number of $\alpha$-approximate minimum quotients (i.e., cut-sets) in connected $k$-hypergraphs with $n$ vertices and $m$ hyperedges is $O(m^{\alpha k} k n^{2\alpha})$. We observe that this is weaker than the Kogan-Krauthgamer bound. We are able to strengthen this bound by observing that if the rank of the hypergraphic polymatroid obtained by contracting $E\setminus \cN'$ is at most $\alpha(c+k-1)=\alpha k$, then the number of vertices in the hypergraph obtained by contracting $E\setminus \cN'$ is at most $\alpha k+1$ and hence, the number of cuts (i.e., non-empty proper vertex subsets) in the contracted hypergraph is $O(2^{\alpha k})$ and consequently, the number of $\alpha$-approximate minimum cuts is $O(2^{\alpha k}kn^{2\alpha})$. 
See Section \ref{section:quotient-bounded-instance} for details. 

In this work we focused on defining the abstract class of quotient-bounded polymatroids and unifying the existing results on mincut computation. A natural question is whether one can derive new results via this notion. The polynomial-time computability of cogirth for various matroid families is an important question and it appears that our definition can capture a few classes that have been studied \cite{Zaslavsky,GGW18} but we leave a formal verification for future work. Polymatroids are more general but less is known about minor-closed families and we hope that our work will give an impetus to discovering tractable families with applications.

\input{relatedwork}

%% file: relatedwork.tex
\subsection{Related Work}
Random contraction has led to many strands of work and it is infeasible to do a thorough job of discussing all of them in this focused work. We give a few pointers.  Our focus is on polynomial-time solvability but adaptations of the contraction approach to design faster algorithms have been studied including the well-known work of Karger and Stein on graphs \cite{KS96}, and Fox, Panigrahi, and Zhang for hypergraphs \cite{fox2023minimum} (also see \cite{KW21}). 
A significant application of random contraction, which is related to the bound on approximate mincuts, are approximation schemes for network unreliability. Starting with Karger's initial breakthrough that obtained the first FPRAS for network unreliability \cite{Kar01}, there have been many papers \cite{Kar17, Kar20, CLP25}, including a recent quasi-polynomial time approximation scheme for hypergraph unreliability \cite{CLP24}. 
Random contraction has also been used to prove several results on multi-criteria mincuts in graphs and hypergraphs---we refer the reader to \cite{Kar16, AMR17, BCX23}.

Random contraction has been used to obtain algorithms for the min $k$-cut problem in graphs, hypergraphs, and constant span hedgegraphs where the goal is to partition the vertex set into $k$ non-empty parts to minimize the number of edges/hyperedges/hedges that cross the parts \cite{KS96, GHLL21, CXY21}. There are some subtleties in extending random contraction approaches for $k$-cut to the setting of polymatroids and this is a topic for future exploration.

There is a very large literature on random sampling and sparsification in graphs, hypergraphs, matroids, and polymatroids. We already mentioned cut and quotient sparsification \cite{BK15,chen2020near,quanrud2024quotient}. Cut sparsification has been generalized to spectral sparsification in graphs \cite{spielman2011spectral,spielman2011graph,batson2012twice} and hypergraphs \cite{KKTY22,Lee23,JLS23,BST19,KKTY21-FOCS, KKTY21-STOC}. We will not attempt to explain these complex results here.

Finally, we mention that the min-quotient problem is of interest in several classes of matroids. A particularly relevant paper is that of Geelen and Kapadia \cite{GeelenKapadia} who showed that girth and cogirth of perturbed graphic matroids can be computed in polynomial-time, and a key tool in their work is Karger's random contraction. Along the way they solve a problem of Barahona and Conforti on computing the cogirth of even cycle matroids \cite{barahona1987construction}. Their work is related to a general conjecture of Gerards, Geelen, and Whittle \cite{geelen2015highly} on polynomial-time solvability of the girth for every \emph{proper} minor closed family of binary matroids (note that Vardy's result shows that the problem is hard for binary matroids). We refer the reader to \cite{GeelenKapadia} for more details.

%% file: preliminary.tex
\section{Preliminaries}\label{section:preliminaries}
We recall that a polymatroid is an integer-valued monotone submodular function with its value on the empty set being $0$. Let $f:2^\mathcal{N}\rightarrow \mathbb{Z}_{\geq 0}$ be a polymatroid. We denote the rank of $f$ by $r:=f(\mathcal{N})$ and the size of the ground set $\mathcal{N}$ by $n:=|\mathcal{N}|$.
For a subset $A\subseteq \cN$, the set function $f_{\backslash A}:2^{\cN\setminus A}\rightarrow \mathbb{Z}_{\geq 0}$ obtained by deleting $A$ is defined by $f_{\backslash A}(S):=f(S)$ for every $S\subseteq \cN\setminus A$, and the set function $f_A:2^{\cN\setminus A}\rightarrow \mathbb{Z}_{\geq 0}$ obtained by contracting $A$ is defined by $f_A(S):=f(S\cup A)-f(A)$ for every $S\subseteq \cN\setminus A$. We note that both deletion and contraction gives a polymatroid again. For ease of notation, for an element $e\in \mathcal{N}$, we let $e$ denote the singleton set $\{e\}$. For every set $S\subseteq \mathcal{N}$, we use $S+e$ and $S-e$ to abbreviate $S\cup \{e\}$ and $S\setminus \{e\}$, respectively.

The following lemma shows basic properties of deletion and contraction.

\begin{restatable}{lemma}{quotientproperty}\label{lemma:quotient-property}
    Let $f:2^{\mathcal{N}}\rightarrow \mathbb{Z}_{\geq 0}$ be a polymatroid and $Q\subseteq \cN$ be a non-empty quotient of $f$. 
    \begin{enumerate}
        \item For every subset $T\subseteq \cN$, $Q\setminus T$ is a quotient of polymatroid $f_{\backslash T}$.
        \item If $Q$ is a minimum non-empty quotient of $f$ with $|Q|>1$, then for every element $e\in Q$, the set $Q-e$ is a minimum non-empty quotient of $f_{\backslash e}$.
        \item If $Q$ is a minimum non-empty quotient of $f$, then for every element $e\not\in Q$, the set $Q$ is a minimum non-empty quotient of $f_e$.
        \item If $|\cN|>1$, then for every element $e\in \cN$, 
        the minimum size of a non-empty quotient of $f_e$ is at least the minimum size of a non-empty quotient of $f$. 
    \end{enumerate}
\end{restatable}

The following lemma gives an upper bound on the number of non-empty quotients in terms of $n$ and $r$.
\begin{restatable}{lemma}{lemmacounting}\label{lemma:quotient-counting}
    Let $f:2^{\mathcal{N}}\rightarrow \mathbb{Z}_{\geq 0}$ be a polymatroid. Then, the number of non-empty quotients of $f$ is at most $n^r$, where $r$ and $n$ denote the rank and size of the ground set, respectively.
\end{restatable}

We defer the proofs of these lemmas to Appendix~\ref{appendix:quotient}.

%% file: nonuniform-random-contraction-3.tex
\section{\MinQuo in $c$-\sqbounded polymatroids}\label{section:non-uniform-contraction}

In this section, we analyze a non-uniform random contraction algorithm for $c$-\sqbounded polymatroid families. 
The following is the main theorem of this section. 
\begin{restatable}{theorem}{theoremrandomcontraction}\label{theorem:random-contraction-algorithm}
    Let $\cF$ be a $c$-\sqbounded family of polymatroids and let $f \in \mathcal{F}$ with rank $r$ and on a ground set of size $n$. There is a random-contraction algorithm that takes $f$ as input (via its evaluation oracle) and runs in polynomial time to return a collection of $n^{O(c)}$ quotients of $f$ and has the following property: for any fixed minimum non-empty quotient $Q$ of $f$, $Q$ is in the output set with probability at least $1/r^{O(c)}$.
\end{restatable}



Running the Algorithm from Theorem~\ref{theorem:random-contraction-algorithm} $r^{O(c)}$ times and returning the minimum quotient produced gives the following.
\begin{corollary}\label{coro:sqbounded-min-quotient-algo}
    Let $\cF$ be a $c$-\sqbounded family of polymatroids. There is a random-contraction algorithm that takes a polymatroid $f\in \cF$ as input and runs in time $(r\cdot n)^{O(c)}$ to return a minimum quotient of $f$ with high probability.
\end{corollary}

Theorem~\ref{theorem:random-contraction-algorithm} also yields an upper bound on the number of minimum quotients.
\begin{corollary}\label{coro:sqbounded-number}
    Let $\cF$ be a $c$-\sqbounded family of polymatroids and $f\in \cF$. The number of minimum non-empty quotients of $f$ is at most $(r\cdot n)^{O(c)}$, where $r$ and $n$ denote the rank and size of the ground set of $f$, respectively.
\end{corollary}

The rest of the section is devoted to proving Theorem \ref{theorem:random-contraction-algorithm}. Our algorithm is a non-uniform random contraction algorithm that is inspired
by the work in \cite{CXY21}.

\subsection{Non-uniform Random Contraction Algorithm}\label{section:non-uniform-algorithm}
We give an overview of the algorithm which is designed to return a set of quotients. Let $f:2^{\cN}\rightarrow \mathbb{Z}_{\geq 0}$ with rank $r$ be a polymatroid. The algorithm first removes all elements $e\in \cN$ with $f(e)=0$. Suppose there exists a "large" element $e\in \cN$ with $f(e)\geq r-c$. 
Contracting $e$ reduces the rank to $\le c$ and we can enumerate all quotients in $f_e$ via brute force in $(n-1)^{c}$ time; note that these are all quotients of $f$ that do \emph{not} contain $e$. We also wish to enumerate all quotients that contain $e$. Since contracting $e$ makes the problem easy (essentially enter a base case), we can also afford to recurse on the polymatroid $f_{\backslash e}$ obtained by deleting $e$ 
If there is no large element we randomly sample an element $e\in \cN$ with probability proportional to $\lambda_e:=r-f(e)-c$ and recurse on the polymatroid obtained by contracting $e$. This is non-uniform random contraction. The pseudocode is given in Algorithm~\ref{algo:pseudocode-nonuniform}. 

\begin{algorithm2e}[h]
\caption{Non-uniform Random Contraction Algorithm}
\label{algo:pseudocode-nonuniform}
\SetKwInput{KwInput}{Input}                
\SetKwInput{KwOutput}{Output}              
\LinesNumbered

\DontPrintSemicolon
  
    \KwInput{Evaluation oracle access to polymatroid $f:2^{\calN}\rightarrow \Z_{\geq 0}$ with $r:=f(\cN)$. 
    }
    \KwOutput{A collection $\cQ$ of quotients of polymatroid $f$.}
    \SetKwFunction{Search}{Search}
    \SetKwProg{Fn}{Function}{:}{}
    \Fn{\Search{$f$}}{
        Remove all elements $e\in \cN$ with $f(e)=0$ from $\cN$. \label{pseudocode-nonuniform-remove1}

            \If{$|\cN|=0$}{
                \KwRet $\emptyset$.
            }
        
            \If{$\exists e\in \cN: f(e)\geq r-c$}{

                Let $e\in \cN$ be an arbitrary element with $f(e)\geq r-c$.\tcp*{$f(e)$ is large.}\label{pseudocode-nonuniform-choice-of-e}
                
                Find the collection $\cQ_e$ of all minimum non-empty quotients without $e$ by brute-force. \label{pseudocode-bruteforce}

                $\mathcal{S}_1\gets \text{\Search}(f_{\backslash e})$.   \label{pseudocode-nonuniform-delete}

                $\mathcal{S}\gets \{S\cup \{e\}:S\in \mathcal{S}_1\}\cup\{e\}$.

                Remove every set $S\in \mathcal{S}$ that is not a minimal non-empty quotient of $f$. \label{pseudocode-nonuniform-remove}

                \tcp*{$\mathcal{S}$ is a collection of quotients that contain $e$.}

                \KwRet $\cQ_e \cup \mathcal{S}$. \label{pseudocode-nonuniform-return2}
            }
            \Else{ \label{pseudocode-nonuniform-else2}
                Set $\lambda_e\gets r-f(e)-c$ for each $e\in \cN$. \label{pseudocode-nonuniform-weight}
        
                Sample an element $e\in \cN$ with probability proportional to $\lambda_e$. \label{pseudocode-nonuniform-contract}
            
                \KwRet $\text{\Search}(f_e)$. \label{pseudocode-nonuniform-recurse} \tcp*{Contract $e$ and recurse}
            }
    }
\end{algorithm2e}


We prove two lemmas that together imply Theorem~\ref{theorem:random-contraction-algorithm}.
The first, and easier lemma, bounds the run-time and number of quotients returned by Algorithm~\ref{algo:pseudocode-nonuniform} and it applies to any polymatroid.

\begin{lemma}\label{lemma:non-uniform-runtime}
    Let $f:2^{\cN} \rightarrow \mathbb{Z}_+$ be an arbitrary polymatroid with rank $r$ and $n = |\cN|$. Algorithm~\ref{algo:pseudocode-nonuniform} on input $f$ returns $n^{c+1}$ quotients in $n^{O(c)}$ time.
\end{lemma}

\begin{proof}
 We first prove the upper bound on the number of quotients returned by induction on $n$. If there is no large element the algorithm does random contraction and reduces $n$ by $1$. Thus the interesting case is when there is a large element. Then the algorithm's output consists of the union of the quotients from the recursive call on $f_{\backslash e}$ (with one element less)  and the quotients from the brute force enumeration in $f_e$. The number of quotients in $f_e$ is at most $(n-1)^c$. Thus we obtain an easy recurrence for $T(n)$, the number of quotients returned by the algorithm, as
 $$T(n) \le (n-1)^c + T(n-1)$$
 with $T(n) \le 1$ for $n \le 1$. Thus, $T(n) \le n^{c+1}$.

 We now analyze the run-time. In each recursive call to $\text{\Search}(f)$, the size of the ground set decreases by at least $1$. At each level, the brute-force search in Line~\ref{pseudocode-bruteforce} is triggered only when $f(e)\geq r-c$, which takes at most $n^{r-f(e)}\leq n^{c}$ time. Sampling takes $O(n)$ time. Thus we obtain a recurrence for the run-time $R(n)$ as: $R(n) \le R(n-1) + O((n-1)^{c})$ with $R(1) = O(1)$ and this yields the desired upper bound.
\end{proof}

\subsection{Success Probability Analysis}\label{section:non-uniform-analysis}
We now analyze the probability of a fixed quotient being in the collection returned by  Algorithm~\ref{algo:pseudocode-nonuniform}.

\begin{lemma}\label{lemma:non-uniform-success}
    Let $\cF$ be a $c$-\sqbounded family of polymatroids and $f\in \cF$ with rank $r$. For any fixed minimum non-empty quotient $Q$ of $f$, the collection of quotients returned by 
    Algorithm~\ref{algo:pseudocode-nonuniform} on input $f$ contains $Q$ with probability at least $1/r^{O(c)}$.
\end{lemma}
\begin{proof}
        For $r\in \Z_+$, let $P(r)$ be the infimum over all $f\in \cF$ with rank $r$ and all minimum non-empty quotient $Q$ of $f$, of the probability that the collection of quotients returned by $\text{\Search}(f)$ contains $Q$. We define $B(r):=\binom{r}{c+1}^{-1}$ if $r\geq c+1$ and $B(r):=1$ otherwise. We prove that $P(r)\geq B(r)$ by induction on $n+r$, where $n$ is the size of the ground set of $f$. 
    Let $f\in \cF$ with rank $r$, ground set $\cN$ with $|\cN|=n$, and a minimum non-empty quotient $Q$. We may assume that $f(e)\ge 1\ \forall e\in \cN$ since all elements $e\in \cN$ with $f(e) = 0$ are not in $Q$ by definition of $Q$ being a quotient and such elements are removed by the algorithm in Step~\ref{pseudocode-nonuniform-remove1}. 
    
    For the base case, we consider $n+r\leq c$, which further implies that $r\leq c$. Then, for every element $e$ in the ground set of $f$, we have that $r-f(e) \leq c$, which implies that Algorithm~\ref{algo:pseudocode-nonuniform} would use a brute-force search to find all minimum quotients, and consequently, the collection returned by Algorithm \ref{algo:pseudocode-nonuniform} contains $Q$ with probability $1$. 

    Next, we prove the induction step. We only consider $r>c$, otherwise Algorithm~\ref{algo:pseudocode-nonuniform} would still use a brute-force search to find all minimum quotients, and return $Q$ with probability $1$. We first assume that there exists an element $e\in \cN$ with $f(e) \ge r-c$. Let $e$ be the element picked by the algorithm in Step~\ref{pseudocode-nonuniform-choice-of-e}.   We have two possibilities: either $e\in Q$ or $e\not\in Q$. 
    
    Suppose that $e\in Q$. Then, either $Q=\{e\}$ or $Q\supsetneq \{e\}$. If $Q=\{e\}$, then the collection $\mathcal{S}$ only contains $\{e\}$ after Step~\ref{pseudocode-nonuniform-remove}. If $Q\supsetneq \{e\}$, then $\{e\}$ is not a quotient and 
    the algorithm recurses on polymatroid $f_{\backslash e}$ in Step~\ref{pseudocode-nonuniform-delete}. Moreover, $|Q|>1$. By Lemma \ref{lemma:quotient-property}, we have that $Q-e$ is a minimum non-empty quotient of the deletion polymatroid $f_{\backslash e}$. Moreover, the rank of $f_{\backslash e}$ at most $r$ and the size of the ground set is at most $n-1$. By induction hypothesis, with probability at least $B(r)$, we have that $Q-e$ is in the collection $\mathcal{S}$ obtained in Step~\ref{pseudocode-nonuniform-delete} and consequently, $Q$ is in  the collection $\cQ_e\cup \mathcal{S}$ returned in Step~\ref{pseudocode-nonuniform-return2}. 
    
    Suppose that $e\not\in Q$. Then, $Q\in \cQ_e$, where $\cQ_e$ is the collection of all minimum quotients that do not contain $e$ as found by the algorithm in Step~\ref{pseudocode-bruteforce}. Consequently, $Q$ is in the collection returned by the algorithm.

    Finally, assume that there does not exist an element $e\in \cN$ with $f(e) \ge r-c$. Then, the algorithm executes the else-clause in Step~\ref{pseudocode-nonuniform-else2}. For each $e\not\in Q$, by Lemma \ref{lemma:quotient-property}, $Q$ is still a minimum non-empty quotient in the contracted polymatroid $f_e$, and moreover, the rank of the contracted polymatroid $f_e$ is $r-f(e)$ and hence, the set $Q$ is in the collection returned by the algorithm executed on $f_e$ with probability at least $P(r-f(e))$. The quotient $Q$ will be in the collection returned by the algorithm if the sampled element $e$ is not from $Q$ and moreover $Q$ is in the collection returned by the algorithm when executed on the contracted polymatroid $f_e$. The rank of the contracted polymatroid $f_e$ is $r-f(e)\leq r-1$. Hence,
    \begin{align}
        P(r) &\geq \frac{1}{\sum\limits_{e\in \cN}\lambda_e}\cdot \sum_{e\in \cN\setminus Q} \lambda_e\cdot P(r-f(e)) \notag \\
        &\geq \frac{1}{\sum\limits_{e\in \cN}(r-f(e)-c)}\cdot \sum_{e\in \cN\setminus Q} (r-f(e)-c)\cdot \binom{r-f(e)}{c+1}^{-1} \notag \ \text{(by induction hypothesis)}\\
        &= \frac{1}{\sum\limits_{e\in \cN}(r-f(e)-c)}\cdot \sum_{e\in \cN\setminus Q} (r-f(e)-c)\cdot \frac{(c+1)!}{\prod_{i=0}^{c}(r-f(e)-i)} \notag \\
        &= \frac{1}{\sum\limits_{e\in \cN}(r-f(e)-c)}\cdot \sum_{e\in \cN\setminus Q} \frac{(c+1)!}{\prod_{i=0}^{c-1}(r-f(e)-i)} \notag \\
        &\geq \frac{1}{\sum\limits_{e\in \cN}(r-f(e)-c)}\cdot \sum_{e\in \cN\setminus Q} \frac{(c+1)!}{\prod_{i=0}^{c-1}(r-1-i)} \ \ \text{(since $f(e)\geq 1\ \forall\ e\in \cN$)} 
        \notag \\
        &= \frac{1}{\sum\limits_{e\in \cN}(r-f(e)-c)}\cdot \sum_{e\in \cN\setminus Q} r\cdot \binom{r}{c+1}^{-1} 
        = r \cdot \binom{r}{c+1}^{-1}\cdot \frac{|\cN\setminus Q|}{\sum\limits_{e\in \cN}(r-f(e)-c)}. \label{inequality:non-uniform-analysis-1}
    \end{align}
    We recall that $\cF$ is a $c$-\sqbounded family of polymatroids and $f\in \cF$. Hence, by definition, we have that 
    \begin{equation}    
    |Q|\leq \frac{\sum_{e\in \cN}(f(e)+c)}{r} \implies |\cN\setminus Q| \geq \frac{\sum_{e\in \cN}(r-f(e)-c)}{r}. \label{ineq:c-sq-bounded}
    \end{equation}
    Substituting this bound for $|\cN\setminus Q|$ in inequality~(\ref{inequality:non-uniform-analysis-1}) gives 
    $$\begin{aligned}
        P(r) & \geq r \cdot \binom{r}{c+1}^{-1}\cdot \frac{|\cN\setminus Q|}{\sum\limits_{e\in \cN}(r-f(e)-c)} \geq \binom{r}{c+1}^{-1}=B(r).
    \end{aligned}$$ 
\end{proof}

We emphasize that the definition of $c$-\sqbounded-polymatroids is used only in inequality (\ref{ineq:c-sq-bounded}) in the analysis.

%% file: approx-random-contraction.tex
\section{\MinQuo in $c$-\qbounded polymatroids}\label{section:approx-algorithm}

In this section, we design a random contraction algorithm to return $(1+\varepsilon)$-approximate quotients in $c$-\qbounded polymatroid families. The following is the main theorem of this section. 





\begin{theorem}\label{theorem:fptas-algorithm}
  Let $\cF$ be a $c$-\qbounded family of polymatroids and $\varepsilon>0$. There exists a randomized algorithm that takes a polymatroid $f\in \cF$ as input (via its evaluation oracle) and runs in polynomial time to return a collection of $n^{O(c)}$ quotients of $f$, where $n$ denotes the size of the ground set of $f$. Moreover, the collection returned by the algorithm contains a quotient $Q$ such that 
  $|Q|\leq (1+\varepsilon)\cdot q$ with probability $1/r^{O(c\cdot\log{(1/\varepsilon)})}$, where $r$ denotes the rank of $f$ and $q$ is the cardinality of a minimum non-empty quotient of $f$.
\end{theorem}


We observe that running the Algorithm from Theorem \ref{theorem:fptas-algorithm} $r^{O(c\cdot\log (1/\varepsilon))}$ times and returning the minimum quotient produced among all runs gives a $(1+\varepsilon)$-approximation minimum quotient with high probability. We formalize it as follows.

\begin{corollary}\label{coro:qbounded-quotient-fptas-algo}
  Let $\cF$ be a $c$-\qbounded family of polymatroids and $\varepsilon>0$. There exists a randomized algorithm that takes a polymatroid $f\in \cF$ as input (via its evaluation oracle) and  
  runs in time $r^{O(c\log{(1/\varepsilon)})}n^{O(c)}$ to return a $(1+\varepsilon)$-approximation minimum quotient of $f$ with high probability, where $r$ and $n$ denote the rank and size of the ground set, respectively.
\end{corollary}

Moreover, if we know the minimum size $q$ of a quotient of $f$, then by setting $\varepsilon<1/q$ and running the Algorithm from Theorem \ref{theorem:fptas-algorithm} $r^{O(c\log (1/\varepsilon))}$ times and returning the minimum quotient produced among all runs leads to an algorithm to find a minimum quotient with high probability with a run-time of $r^{O(c\log{q})}n^{O(c)}$. If we do not know the minimum size $q$ of a quotient of $f$, then we can find a min-sized quotient in time $r^{O(c\log{q})}n^{O(c)}$ as follows: run the previously mentioned algorithm for each guess $q'=1, 2, 3, \ldots, n$ of $q$ and terminate after the least value of $q'$ for which the algorithm returns a quotient of size at most $q'$. Thus, we have a quasi-polynomial time algorithm to find a minimum quotient of polymatroids in $c$-\qbounded families as summarized in the following corollary. 

\begin{corollary}\label{coro:qbounded-quotient-algo-quasi-polytime}
    Let $\cF$ be a $c$-\qbounded family of polymatroids. There exists a randomized algorithm that takes a polymatroid $f\in \cF$ as input (via its evaluation oracle) and runs in time $r^{O(c\log{q})}n^{O(c)}$ to return a minimum quotient of $f$ with high probability, where $r$, $n$, and $q$ denote the rank, size of the ground set, and size of a minimum quotient of $f$, respectively.
\end{corollary}

The algorithm from Theorem \ref{theorem:fptas-algorithm} also provides an upper bound on the number of minimum quotients as follows, which needs a slightly more careful analysis. We defer the proof of the following corollary to Appendix~\ref{appendix:gkp-counting-bound}.
\begin{corollary}\label{coro:qbounded-number-1}
    Let $\cF$ be a $c$-\qbounded family of polymatroids and $f\in \cF$. The number of minimum non-empty quotients of $f$ is $r^{O(c\log{q})}n^{O(c)}$, where $r$, $n$, and $q$ denote the rank, size of the ground set, and size of a minimum quotient of $f$, respectively.
\end{corollary}

The rest of the section is devoted to proving Theorem \ref{theorem:fptas-algorithm}. Our algorithm generalizes the random contraction algorithm due to Ghaffari, Karger, and Panigrahi \cite{GKP17} to find $(1+\varepsilon)$-approximate min-cuts in hedgegraphs. 
In Section~\ref{subsection:approx-algorithm}, we present the random contraction algorithm. In Section~\ref{subsection:approx-prob}, we analyze the success probability for polymatroids in $c$-\qbounded families and prove Theorem~\ref{theorem:fptas-algorithm}.

\subsection{Randomized Algorithm}\label{subsection:approx-algorithm}

We first recall the randomized algorithm of Ghaffari, Karger, and Panigrahi \cite{GKP17} to find $(1+\varepsilon)$-approximate min-cuts in hedgegraphs. Let $G=(V,E)$ be the given hedgegraph. The hedges are categorized into three groups according to the number of vertices contained: a hedge $e\in E$ is  \emph{large} if $|e|\geq |V|/2$, \emph{moderate} if $|V|/2> |e| \geq |V|/4$, and \emph{small} if $|e| < |V|/4$, where $|e|$ is the number of vertices in $e$. The algorithm iteratively performs the following steps until there are only two vertices in $G$: If none of the hedges in $E$ is large, then a hedge is chosen uniformly at random and contracted. Otherwise, the algorithm finds all hedges $e\in E$ that are large or moderate, i.e.,
$$L:=\{e\in E: |e|\geq \frac{|V|}{4}\}.$$
Then, the algorithm branches with equal probability to one of the following: either all these hedges are removed from $G$ and added to the output cut, or one of these hedges in $L$ is contracted uniformly at random.

The analysis in \cite{GKP17} showed that the uniform random contraction behaves well when none of the hedges is large. In the presence of large hedges, we compare $L$ to the minimum cut $C$. If the minimum cut contains at least a $\frac{1}{1+\varepsilon}$ fraction of hedges in $L$, i.e., $|L\cap C|\geq \frac{1}{1+\varepsilon}\cdot |L|$, we may add all hedges in $L$ to the output cut. Otherwise, $|L\cap C|<\frac{1}{1+\varepsilon}\cdot |L|$, which implies that the probability of a randomly chosen hedge from $L$ belonging to $C$ is quite low. We emphasize that there is no need to know how many large or moderate hedges are in the minimum cut $C$, since the algorithm always chooses the right branch with a probability of $\frac{1}{2}$.

We now extend this approach to $c$-\qbounded polymatroid families. Let $\cF$ be a $c$-\qbounded family of polymatroids and $f\in \cF$ be a polymatroid of rank $r$. The algorithm first removes all elements $e\in \cN$ with $f(e)=0$. If $r\leq 4c$, by Lemma~\ref{lemma:quotient-counting}, there are at most $n^{4c}$ non-empty quotients of $f$ and the algorithm does a brute-force search to find all minimum quotients. If $n\leq 4c$, there are at most $2^n\leq n^{4c}$ quotients of $f$ and a brute-force search is applied again. We now assume that $n>4c$ and $r>4c$. For every element $e\in \cN$, we call it \emph{heavy} if $f(e)\geq \frac{1}{2c}\cdot r$, \emph{moderate} if $\frac{1}{2c}\cdot r > f(e)\geq \frac{1}{4c}\cdot r$, and \emph{light} if $f(e)<\frac{1}{4c}\cdot r$. The algorithm checks whether there exists a heavy element. If none of the elements are heavy, then we recurse on the polymatroid obtained by contracting a uniformly random element. Otherwise, the algorithm finds the set $L$ of elements that are heavy or moderate and branches with equal probability to one of the following: either recurse on the polymatroid obtained by deleting $L$ to obtain a collection of quotients and add $L$ to all of them (and perform a clean-up of the resulting collection to ensure that all of them are quotients), or recurse on the polymatroid obtained by contracting a uniformly random element from $L$. The pseudocode is given in Algorithm~\ref{algo:pseudocode-gkp}. Lemmas~\ref{lemma:gkp-runtime} and \ref{lemma:gkp-success-probability} that are shown below together imply Theorem~\ref{theorem:fptas-algorithm}.

\begin{algorithm2e}[h]
\caption{Approximation Algorithm for \MinQuo}
\label{algo:pseudocode-gkp}
\SetKwInput{KwInput}{Input}                
\SetKwInput{KwOutput}{Output}              
\LinesNumbered

\DontPrintSemicolon
  
    \KwInput{evaluation oracle access to polymatroid $f:2^{\calN}\rightarrow \Z_{\geq 0}$ with $r=f(\cN)$. 
    }
    \KwOutput{a collection $\cQ$ of quotients of polymatroid $f$.}
    \SetKwFunction{Search}{Search}
    \SetKwProg{Fn}{Function}{:}{}
    
    \Fn{\Search{$f$}}{
        Remove all elements $e\in \cN$ with $f(e)=0$. \label{pseudocode-gkp-deleting}
        
        \If{$r\leq 4c$ or $n\leq 4c$}{
            Find the collection $\cQ$ of all minimum non-empty quotients by brute-force. \label{pseudocode-gkp-bruteforce}

            \KwRet $\cQ$.
        }

        $L_1\gets \{e\in \cN: f(e)\geq \frac{1}{2c}\cdot r\}$. \tcp*{heavy elements}
        
        $L_2\gets \{e\in \cN: \frac{1}{2c}\cdot r > f(e)\geq \frac{1}{4c}\cdot r\}$. \tcp*{moderate elements}

        \If{$L_1=\emptyset$}{ \label{pseudocode-gkp-if}
        
            Sample element $e\in \cN$ uniformly at random.

            \KwRet $\text{\Search}(f_e)$.
        }
        \Else{ \label{pseudocode-gkp-else}
        
            $L\gets L_1\cup L_2$.

            Sample $X\in \{0,1\}$ uniformly at random. \tcp*{branching step} \label{pseudocode-gkp-branching}

            \If{$X=0$}{ \label{pseudocode-gkp-branch1}

                \If{$f(\cN\setminus L)<r$}{
                
                    $\cQ\gets \{L\}$.  \tcp*{$L$ is a superset of quotient.} \label{pseudocode-gkp-singleton}
                }
                \Else{
                
                    $\cQ\gets \text{\Search}(f_{\backslash L})$.

                    Add all elements in $L$ into every set in $\cQ$.
                }
                
                \For{$Q\in \cQ$}{
                
                    $Q\gets \cN\setminus (\Span(\cN\setminus Q))$. \tcp*{This step ensures that $Q$ is a quotient.} \label{pseudocode-gkp-clean}
                }

                \KwRet $\cQ$.
            }

            \Else{ \label{pseudocode-gkp-branch2}
            
                Sample element $e\in L$ uniformly at random.

                \KwRet $\text{\Search}(f_e)$.
            }
        }
    }
\end{algorithm2e}

The following lemma bounds the run-time and the number of quotients returned by Algorithm~\ref{algo:pseudocode-gkp}.

\begin{lemma}\label{lemma:gkp-runtime}
    Let $\cF$ be a $c$-\qbounded family of polymatroids and $f\in \cF$. 
    Algorithm~\ref{algo:pseudocode-gkp} on input $f$ returns $n^{4c}$ quotients in $n^{O(c)}$ time, where $n$ denotes the size of the ground set of $f$.
\end{lemma}
\begin{proof}
    We first prove the upper bound on the number of quotients returned by induction on $n$. For the base case of $n\leq 4c$, the number of minimum quotients is at most $2^n\leq n^{4c}$ and Algorithm~\ref{algo:pseudocode-gkp} would use a brute-force search to find them in Step~\ref{pseudocode-gkp-bruteforce}.

    We now prove the induction step. If $r\leq 4c$, then by Lemma~\ref{lemma:quotient-counting}, there are at most $n^{4c}$ minimum non-empty quotients. Algorithm~\ref{algo:pseudocode-gkp} would use a brute-force search to find them in Step~\ref{pseudocode-gkp-bruteforce}. We now assume that $r>4c$. First, suppose there is no heavy element, i.e., $L_1=\emptyset$. The algorithm executes the if-clause in Step~\ref{pseudocode-gkp-if}. By induction hypothesis, the number of quotients returned is at most
    $$(n-1)^{4c}<n^{4c}.$$

    Next, assume that $L_1\neq \emptyset$. The algorithm executes the else-clause in Step~\ref{pseudocode-gkp-else} and branches to $X=0$ or $X=1$. For the case of $X=0$, the algorithm executes the if-clause in Step~\ref{pseudocode-gkp-branch1}; it either returns $\{L\}$ or it recurses on $f_{\backslash L}$ to obtain a collection $\cQ$, adds $L$ to each set in $\cQ$, and a subset of each such set that is a quotient. By induction hypothesis, the number of quotients returned is at most
    $$\max\{1, (n-1)^{4c}\}<n^{4c}.$$
    For the case of $X=1$, the algorithm executes the else-clause in Step~\ref{pseudocode-gkp-branch2}, where it recurses on $f_e$ for a uniform random element $e\in L$. By induction hypothesis, the number of quotients returned is at most
    $$(n-1)^{4c}<n^{4c}.$$

    We now analyze the runtime. At each recursive call to $\Search(f)$, the size of the ground set decreases by at least $1$, so the recursion depth is at most $n$. At each level, the brute-force search in Line~\ref{pseudocode-gkp-bruteforce} is triggered only when $r\leq 4c$ or $n\leq 4c$, which takes at most $n^{4c}$ time. Thus, the overall runtime is $n^{O(c)}$.
\end{proof}

\subsection{Success Probability Analysis}\label{subsection:approx-prob}

We now analyze the probability that a $(1+\varepsilon)$-approximation minimum quotient is contained in the collection returned by Algorithm~\ref{algo:pseudocode-gkp}. We first show that a single execution of \text{\Search}$(f)$ performs $O(c\cdot \log r)$ branching steps in Step~\ref{pseudocode-gkp-branching}.

\begin{lemma}\label{lemma:gkp-depth}
Let $\cF$ be a $c$-\qbounded family of polymatroids and $f\in \cF$ with $r$ being the rank of $f$.
    Then, a single execution of 
    Algorithm~\ref{algo:pseudocode-gkp} on input $f$ 
    performs at most $4c\cdot \log r$ branching steps in Step~\ref{pseudocode-gkp-branching}.
\end{lemma}
\begin{proof}
    We prove by induction on $r$. For the base case of $r\leq 4c$, Algorithm~\ref{algo:pseudocode-gkp} will use a brute-force search and return the collection of minimum quotients. Hence, the lemma holds since there is no recursion.

    We now prove the induction step. If $n\leq 4c$, then the algorithm will again use a brute-force search in Step~\ref{pseudocode-gkp-bruteforce} and no recursion occurs. We may now assume that $n>4c$. If there is no heavy element in $\cN$, i.e., $L_1=\emptyset$, then Algorithm~\ref{algo:pseudocode-gkp} will execute the if-clause in Step~\ref{pseudocode-gkp-if}, where it contracts an element $e\in \cN$ and recurses on the minor $f_e$ of rank strictly smaller than $r$. By induction hypothesis, the execution of Algorithm \ref{algo:pseudocode-gkp} on input $f_e$ performs at most $4c\log{(r-1)}$ branching steps and hence, the execution of Algorithm \ref{algo:pseudocode-gkp} on input $f$ performs at most 
    $4c\cdot \log (r-1) < 4c\cdot \log r$ 
    branching steps. 
    For the case where $L_1$ is non-empty, the algorithm executes the else-clause in Step~\ref{pseudocode-gkp-else} and branches to $X=0$ or $X=1$. We consider the two branches separately:
    \begin{enumerate}
        \item $X=0$: the algorithm executes the if-clause in Step~\ref{pseudocode-gkp-branch1}, where it sets $\cQ=\{L\}$ or recurses on $f_{\backslash L}$. We observe that every remaining element $e\in \cN\setminus L$ is light, i.e., $f(e)<\frac{1}{4c}\cdot r$. This implies that the next branching step can be performed only when the rank has decreased to $r'<r/2$, at which point there may exist heavy elements. According to the inductive hypothesis, the total number of branching steps performed is at most
        $$1 + 4c\cdot \log r' < 1 + 4c\cdot \log (r/2) \leq 4c\cdot \log r.$$

        \item $X=1$: the algorithm executes the else-clause in Step~\ref{pseudocode-gkp-branch2}, where it recurses on $f_e$ for a random element $e\in L$. Since $f(e)\geq \frac{1}{4c}\cdot r$, the rank of $f_e$ is at most $r'\leq r - \frac{1}{4c}\cdot r$. By the inductive hypothesis, the total number of branching steps performed is at most
        $$1 + 4c\cdot \log r' \leq 1 + 4c\cdot \log \left(r-\frac{r}{4c}\right)<4c\cdot \log r,$$
        where the last inequality holds since $1+4x\cdot \log\left(1-\frac{1}{4x}\right)<0$ for every $x>1$.
    \end{enumerate}
\end{proof}


We now analyze the probability of a $(1+\varepsilon)$-approximation minimum quotient being in the collection returned by the algorithm.

\begin{lemma}\label{lemma:gkp-success-probability}
  Let $\cF$ be a $c$-\qbounded family of polymatroids and $f\in \cF$ with $r$ being the rank of $f$ and $q$ being the cardinality of a minimum non-empty quotient. For every constant $\varepsilon>0$, the collection of quotients returned by 
    Algorithm~\ref{algo:pseudocode-gkp} on input $f$ contains a quotient of size at most $(1+\varepsilon)\cdot q$ with probability at least $1/r^{O(c\log{(1/\varepsilon)})}$. 
  \end{lemma}

\begin{proof}
  For every integer $r\geq 1$ and $\ell\geq 0$, let $P(r, \ell)$ be the infimum over all $f\in \cF$ with rank $r$ on which \text{\Search}$(f)$ performs at most $\ell$ branching steps, of the probability that the collection of quotients returned by \text{\Search}$(f)$ contains a quotient $Q'$ with $|Q'|\leq (1+\varepsilon)\cdot q$. By Lemma~\ref{lemma:gkp-depth}, we may assume that $\ell \leq 4c\cdot \log r$. We now prove that
    \begin{align}\label{inequality:gkp-analysis-probability}
        P(r, \ell)\geq \left(\frac{1}{r}\right)^{4c} \cdot \left(\frac{\varepsilon}{2(1+\varepsilon)}\right)^{\ell}    
    \end{align}
    by induction on $r+\ell$. We may also assume that every element $e\in \cN$ has $f(e)\geq 1$ since all elements $e\in \cN$ with $f(e)=0$ do not belong to any non-empty quotient, and such elements are removed by the algorithm in Step~\ref{pseudocode-gkp-deleting}.
    
    For the base case, we consider $r+\ell\leq 4c$. Algorithm~\ref{algo:pseudocode-gkp} performs a brute-force search to find minimum quotients of $f$ in Step~\ref{pseudocode-gkp-bruteforce} since $r\leq 4c$. Consequently, the collection returned by Algorithm~\ref{algo:pseudocode-gkp} contains a quotient of size $q$
 with probability $1$.

    We now prove the induction step where $r+\ell > 4c$. If $r\leq 4c$, Algorithm~\ref{algo:pseudocode-gkp} would terminate in Step~\ref{pseudocode-gkp-bruteforce} via brute-force search with no branching steps, contradicting $r+\ell > 4c$. Thus, $r>4c$. Moreover, we recall that $\cF$ is a $c$-\qbounded family of polymatroids and $f\in \cF$. This implies that
    \begin{align}
        q \leq c\cdot \frac{\sum_{e\in \cN}f(e)}{r}. \label{inequality:gkp-quotient-bound}
    \end{align}

    Let $Q$ be an arbitrary minimum non-empty quotient of $f$ with cardinality $q$.
    We consider the following three cases:
    \begin{enumerate}
    \item $L_1=\emptyset$: In this case, all elements are light or moderate, i.e., $f(e)< \frac{r}{2c}$.  We have
        \begin{align}
            \sum_{e \in \cN\setminus Q} f(e) &= \sum_{e \in \cN} f(e) - \sum_{e\in Q} f(e) \notag \\
            &\geq \sum_{e\in \cN} f(e) - \frac{r}{2c}\cdot |Q| \ \ \text{(since $f(e)<\frac{r}{2c}$)} \notag \\
            &\geq \frac{r}{c}\cdot |Q| - \frac{r}{2c}\cdot |Q| \ \ \text{(by inequality~(\ref{inequality:gkp-quotient-bound}))} \notag \\
            &= \frac{r}{2c}\cdot q. \label{inequality:gkp-analysis-smallcase}
        \end{align}
        This further shows that
        $$\frac{r}{2c}\cdot |Q| \leq \sum_{e\in \cN\setminus Q}f(e)< \frac{r}{2c}\cdot |\cN\setminus Q|,$$
        which implies that $|Q|< |\cN\setminus Q|$ and $|Q|< |\cN|/2$. We observe that Algorithm~\ref{algo:pseudocode-gkp} will execute the if-clause in Step~\ref{pseudocode-gkp-if}. For each $e\not\in Q$, by Lemma~\ref{lemma:quotient-property}, $Q$ is still a minimum non-empty quotient in the contracted polymatroid $f_e$, and moreover, the rank of the contracted polymatroid $f_e$ is $r-f(e)$ and hence, a quotient $Q'$ with $|Q'|\leq (1+\varepsilon)\cdot |Q|$ is in the collection returned by the algorithm executed on $f_e$ with probability at least $P(r-f(e),\ell)$. If the sampled element $e$ is not from $Q$ and moreover, a quotient $Q'$ with $|Q'|\leq (1+\varepsilon)\cdot |Q|$ is in the collection returned by the algorithm when executed on the contracted polymatroid $f_e$, then $Q'$ would also be contained in the collection returned by the algorithm. Therefore, 
        $$\begin{aligned}
            P(r, \ell) &\geq\frac{1}{|\cN|}\cdot \sum_{e\in \cN\setminus Q} P(r-f(e), \ell) \\
            &\geq \frac{1}{|\cN|} \cdot \sum_{e\in \cN\setminus Q} \frac{1}{(r-f(e))^{4c}} \cdot \left(\frac{\varepsilon}{2(1+\varepsilon)}\right)^{\ell} \ \ \text{(by induction hypothesis)} \\
            &\geq \frac{1}{|\cN|}\cdot \frac{|\cN\setminus Q|}{\left(r-\frac{\sum_{e\in \cN\setminus Q}f(e)}{|\cN\setminus Q|}\right)^{4c}}\cdot \left(\frac{\varepsilon}{2(1+\varepsilon)}\right)^{\ell} \ \ \text{(by convexity of the function $x^{-4c}$)} \\
            &\geq \frac{1}{|\cN|}\cdot \frac{|\cN\setminus Q|}{\left(r-\frac{r}{2c}\cdot \frac{|Q|}{|\cN\setminus Q|}\right)^{4c}}\cdot \left(\frac{\varepsilon}{2(1+\varepsilon)}\right)^{\ell} \ \ \text{(by inequality~(\ref{inequality:gkp-analysis-smallcase}))} \\
            &\geq \frac{1}{|\cN|}\cdot \frac{|\cN\setminus Q|}{\left(r-\frac{r}{2c}\cdot \frac{|Q|}{|\cN|}\right)^{4c}}\cdot \left(\frac{\varepsilon}{2(1+\varepsilon)}\right)^{\ell} \\
            &= \frac{1}{r^{4c}}\cdot \left(\frac{\varepsilon}{2(1+\varepsilon)}\right)^{\ell} \cdot \frac{1-\frac{|Q|}{|\cN|}}{\left(1-\frac{1}{2c}\cdot \frac{|Q|}{|\cN|}\right)^{4c}} \\
            &\geq \frac{1}{r^{4c}}\cdot \left(\frac{\varepsilon}{2(1+\varepsilon)}\right)^{\ell},
        \end{aligned}$$
        where the last inequality is because $\frac{|Q|}{|\cN|}\leq \frac{1}{2}$ and $(1-\frac{1}{2c}\cdot x)^{4c}\leq 1-x$ for every $x\in [0, 1/2)$.
        
        \item $L_1\neq \emptyset$ and $|L\setminus Q|< \frac{\varepsilon}{1+\varepsilon}\cdot |L|$: In this case, Algorithm~\ref{algo:pseudocode-gkp} branches to $X=0$ with probability $\frac{1}{2}$. Since $|L\setminus Q|<\frac{\varepsilon}{1+\varepsilon}\cdot |L|$, we have that $|L\setminus Q|<\varepsilon\cdot |L\cap Q|$ and $|L|<(1+\varepsilon)|Q|$. Suppose $Q$ is a subset of $L$ or $f(\cN\setminus L)<r$. Then,  Algorithm~\ref{algo:pseudocode-gkp} would execute Step~\ref{pseudocode-gkp-singleton} and perform a clean-up in Step~\ref{pseudocode-gkp-clean}. Since $|L\setminus Q|<\frac{\varepsilon}{1+\varepsilon}\cdot |L|$ and $|L|<(1+\varepsilon)\cdot |Q|$, there is a quotient of size at most $(1+\varepsilon)\cdot |Q|$ in the collection returned by the algorithm. Next, suppose $Q$ is not a subset of $L$ and $f(\cN\setminus L)=r$. This implies that $|Q\setminus L|>0$. Algorithm~\ref{algo:pseudocode-gkp} would recurse on polymatroid $f_{\backslash L}$ to obtain a collection of quotients of $f_{\backslash L}$ and add all elements in $L$ to every quotient in the collection. By Lemma~\ref{lemma:quotient-property} part (1), $Q\setminus L$ is a non-empty quotient of $f_{\backslash L}$.
          This implies that the probability that a quotient $Q'$ with $|Q'|\leq (1+\varepsilon)\cdot |Q\setminus L|$ is in the collection returned by the algorithm when executed on polymatroid $f_{\backslash L}$ is at least $P(r, \ell-1)$. Moreover, if such a quotient $Q'$ is contained in the collection returned by the recursive call, then a quotient of size at most $|Q'|+|L|$ is in the output. We bound its size as follows:
        \begin{align*}
            |Q'|+|L|
            &\leq (1+\varepsilon)\cdot |Q\setminus L|+|L|\\
            &=|Q\setminus L|+|L\cap Q|+\varepsilon|Q\setminus L|+|L\setminus Q|\\
            &=|Q|+\varepsilon|Q\setminus L | + |L\setminus Q|\\
            &\le |Q|+\varepsilon|Q\setminus L | + \varepsilon|L\cap Q|\\
            &= (1+\varepsilon)\cdot |Q|.
        \end{align*}
Therefore,
        $$\begin{aligned}
            P(r, \ell) &\geq \frac{1}{2}\cdot P(r, \ell -1) \\
            &\geq \frac{1}{2}\cdot \left(\frac{1}{r}\right)^{4c} \cdot \left(\frac{\varepsilon}{2(1+\varepsilon)}\right)^{\ell-1} \ \ \text{(by induction hypothesis)}\\
            &\geq \left(\frac{1}{r}\right)^{4c} \cdot \left(\frac{\varepsilon}{2(1+\varepsilon)}\right)^{\ell}.
        \end{aligned}$$

        \item $L_1\neq \emptyset$ and $|L\setminus Q|\geq \frac{\varepsilon}{1+\varepsilon}\cdot |L|$: In this case, Algorithm~\ref{algo:pseudocode-gkp} branches to $X=1$ with probability $\frac{1}{2}$. Algorithm~\ref{algo:pseudocode-gkp} contracts an element $e\in L$, which is sampled uniformly at random. Thus, the probability that $e\not\in Q$ equals $\frac{|L\setminus Q|}{|L|}$. Since $e$ is heavy or moderate, $f(e)\geq \frac{r}{4c}$, which implies that contracting element $e$ reduces the rank by at least $\frac{r}{4c}$. By Lemma \ref{lemma:quotient-property} part (3), $Q$ is still a minimum non-empty quotient of $f_e$. Hence, a quotient $Q'$ with $|Q'|\leq (1+\varepsilon)\cdot |Q|$ is in the collection returned by the algorithm executed on $f_e$ with probability at least $P(r-\frac{r}{4c},\ell)$. If the sampled element $e$ is not from $Q$ and moreover a quotient $Q'$ with $|Q'|\leq (1+\varepsilon)\cdot |Q|$ is in the collection returned by the algorithm when executed on the contracted polymatroid $f_e$, then $Q'$ would be contained in the collection returned by the algorithm. Therefore,
        $$\begin{aligned}
            P(r, \ell) & \geq \frac{1}{2}\cdot \frac{|L\setminus Q|}{|L|} \cdot P\left(r-\frac{r}{4c}, \ell-1\right) \\
            &\geq \frac{1}{2}\cdot \frac{|L\setminus Q|}{|L|}\cdot \left(\frac{1}{\left(1-\frac{1}{4c}\right)r}\right)^{4c} \cdot \left(\frac{\varepsilon}{2(1+\varepsilon)}\right)^{\ell-1} \ \ \text{(by induction hypothesis)}\\
            &\geq \frac{1}{2}\cdot \frac{\varepsilon}{1+\varepsilon}\cdot \left(\frac{1}{\left(1-\frac{1}{4c}\right)r}\right)^{4c} \cdot \left(\frac{\varepsilon}{2(1+\varepsilon)}\right)^{\ell-1} \ \ \text{(since $|L\setminus Q|\geq \frac{\varepsilon}{1+\varepsilon}\cdot |L|$)}\\
            &\geq \left(\frac{1}{r}\right)^{4c} \cdot \left(\frac{\varepsilon}{2(1+\varepsilon)}\right)^{\ell}.
        \end{aligned}$$
    \end{enumerate}
    Thus, inequality~(\ref{inequality:gkp-analysis-probability}) holds for every polymatroid $f\in \cF$. This implies that the collection of quotients returned by \text{\Search}$(f)$ contains a $(1+\varepsilon)$-approximation minimum non-empty quotient with probability at least
    $$P(r, \ell) \geq \left(\frac{1}{r}\right)^{4c} \cdot \left(\frac{\varepsilon}{2(1+\varepsilon)}\right)^{\ell}\geq \left(\frac{1}{r}\right)^{4c} \cdot \left(\frac{\varepsilon}{2(1+\varepsilon)}\right)^{4c\log r}=\left(\frac{1}{r}\right)^{O(c\log(1/\varepsilon))}.$$

\end{proof}

%% file: parallel-random-contraction.tex
\section{FPT of \MinQuo in $c$-\qbounded polymatroids}\label{section:parallel-contraction}
In this section, we design a fixed-parameter tractable algorithm for \MinQuo in $c$-\qbounded polymatroid families when parameterized by the solution size. The following is the main theorem of this section. 
\begin{theorem}\label{theorem:fpt-algorithm}
  Let $\cF$ be a $c$-\qbounded family of polymatroids. There exists a randomized algorithm that takes a polymatroid $f\in \cF$ as input (via its evaluation oracle) and runs in polynomial time to return a collection of $n^{O(c)}$ quotients of $f$, where $n$ denotes the size of the ground set of $f$. Moreover, for every fixed minimum non-empty quotient $Q$, the collection of quotients returned by the algorithm contains $Q$ with probability $\frac{1}{2^q \cdot r^{O(c\log c)}}$, where $r$ and $q$ denote the rank of the ground set and size of a minimum quotient of $f$, respectively.
\end{theorem}


We observe that Theorem \ref{theorem:fpt-algorithm} leads to a fixed-parameter algorithm to find a minimum quotient of a polymatroid from a $c$-\qbounded family when parameterized by solution size for every fixed constant $c$: running the Algorithm from Theorem \ref{theorem:fpt-algorithm} $2^q \cdot r^{O(c\log c)}$ times and returning the minimum quotient produced among all runs gives a minimum quotient with high probability. If we do not know the minimum size $q$ of a quotient of $f$, then we can find a min-sized quotient in time $2^q \cdot r^{O(c\log{c})}\cdot n^{O(c)}$ as follows: run the previously mentioned algorithm for each guess $q'=1, 2, 3, \ldots, n$ of $q$ and terminate after the least value of $q'$ for which the algorithm returns a quotient of size at most $q'$. This leads to the corollary stated below. 

\begin{corollary}\label{coro:qbounded-algo-fpt}
    Let $\cF$ be a $c$-\qbounded family of polymatroids. There exists a randomized algorithm that takes a polymatroid $f\in \cF$ as input (via its evaluation oracle) and runs in time $2^q\cdot r^{O(c\log c)} \cdot n^{O(c)}$ to return a minimum quotient of $f$ with high probability, where $r$, $n$, and $q$ denote the rank, size of the ground set, and size of a minimum quotient of $f$, respectively.
\end{corollary}

Theorem~\ref{theorem:fpt-algorithm} also provides an upper bound on the number of minimum quotients. We formalize it as follows.
\begin{corollary}\label{coro:qbounded-number-2}
    Let $\cF$ be a $c$-\qbounded family of polymatroids and $f\in \cF$ be a polymatroid. The number of minimum non-empty quotients of $f$ is at most $2^q\cdot r^{O(c\log c)} \cdot n^{O(c)}$, where $r$, $n$, and $q$ denote the rank, size of the ground set, and size of a minimum quotient of $f$, respectively.
\end{corollary}

The rest of this section is devoted to proving Theorem \ref{theorem:fpt-algorithm}. 
Our algorithm generalizes the random contraction based fixed-parameter algorithm for \HedgeMinCut due to Fomin, Golovach, Korhonen, Lokshtanov, and Saurabh \cite{fomin2025fixed}. 
In Section~\ref{subsection:fpt-algorithm}, we present the random contraction algorithm. In Section~\ref{subsection:fpt-prob}, we analyze the success probability for $c$-\qbounded polymatroid families and prove Theorem~\ref{theorem:fpt-algorithm}.

\subsection{Randomized Algorithm}\label{subsection:fpt-algorithm}

We first recall the FPT algorithm of Fomin et al. \cite{fomin2025fixed}. Let $G=(V,E)$ be the given hedgegraph and $\ell$ be the size of the minimum cut. A hedge $e\in E$ is \emph{large} if it contains at least $|V|/8$ vertices, i.e., $|e|\geq |V|/8$ and \emph{small} otherwise. If there exists a large hedge $e\in E$, the algorithm guesses whether $e$ belongs to the optimum cut (with non-uniform probability that depends on the number of large hedges): either hedge $e$ is added to the output cut, in which case we remove it from $G$ and decrease the parameter $\ell$ by one, or hedge $e$ is not added to the output cut, in which case we contract $e$ and the number of vertices in $G$ decreases by at least $|V|/8$. If all hedges are small, the algorithm contracts every hedge $e\in E$ with probability $1/\ell$ independently. The algorithm repeats this process until the number of vertices is small, at which point we use a brute-force search to find the minimum cut. It was shown in \cite{fomin2025fixed} that the algorithm finds the minimum cut with probability $|V|^{-O(1)}\cdot \binom{O(\log |V|)+\ell}{\ell}^{-1}$, which leads to an FPT algorithm parameterized by $\ell$.

We now extend this framework to \MinQuo. Let $\cF$ be a $c$-\qbounded polymatroid family and $f\in \cF$ be a polymatroid of rank $r$. Let $q>0$ be the size of a minimum non-empty quotient of $f$. The algorithm first removes all elements $e\in \cN$ with $f(e)=0$. If $r\leq 2c$, by Lemma~\ref{lemma:quotient-counting}, there are at most $n^{2c}$ non-empty quotients of $f$ and the algorithm does a brute-force search to find all minimum quotients. If $q\leq 2c$, there are at most $\binom{n}{2c}\leq n^{2c}$ subsets of $\cN$ of size $q$ and a brute-force search is applied again. We now assume that $r>2c$ and $q>2c$. An element $e\in \cN$ is \emph{heavy} if $f(e)\geq \frac{1}{2c}\cdot r$, and \emph{light} otherwise. Then, the algorithm checks whether there exists a heavy element. If there exists a heavy element $e\in \cN$, then it branches to one of the following: with probability $\frac{q}{\lfloor 2c\cdot \log r\rfloor+q}$, the algorithm decreases $q$ by one, recurses on the deleted polymatroid $f_{\backslash e}$, and adds $e$ to all quotients returned by the recursion; with probability $\frac{\lfloor 2c\cdot \log r\rfloor}{\lfloor 2c\cdot \log r\rfloor+q}$, the algorithm recurses on the contracted polymatroid $f_e$. If all elements are light, the algorithm samples every element $e\in \cN$ with probability $\frac{1}{q}$ independently and recurses on the polymatroid obtained by contracting the collection of sampled elements. The pseudocode is given in Algorithm~\ref{algo:pseudocode-fpt}. Lemmas~\ref{lemma:fpt-runtime} and \ref{lemma:fpt-success-probability} that are shown below together imply Theorem~\ref{theorem:fpt-algorithm}.

\begin{algorithm2e}[h]
\caption{Fixed-Parameter Tractable Algorithm}
\label{algo:pseudocode-fpt}
\SetKwInput{KwInput}{Input}                
\SetKwInput{KwOutput}{Output}              
\LinesNumbered

\DontPrintSemicolon
  
    \KwInput{evaluation oracle access to polymatroid $f:2^{\cN}\rightarrow \Z_{\geq 0}$ with $r=f(\cN)$, minimum non-empty quotient size $q>0$.
    }
    \KwOutput{a collection $\cQ$ of quotients of polymatroid $f$.}
    \SetKwFunction{Search}{Search}
    \SetKwProg{Fn}{Function}{:}{}
    
    \Fn{\Search{$f, q$}}{
        Remove all elements $e\in \cN$ with $f(e)=0$. \label{pseudocode-fpt-deleting}
        
        \If{$r\leq 2c$ or $q\leq 2c$}{
            Use a brute-force search to find all minimum non-empty quotients $\cQ$ of $f$. \label{pseudocode-fpt-bruteforce}
            
            If $\cQ$ is non-empty, then return $\cQ$, else return FAIL.
        }

        \If{$\exists \ e \in \cN$ with $f(e)\geq \frac{1}{2c}\cdot r$}{ \label{pseudocode-fpt-if}
            Let $e\in \cN$ be an arbitrary element with $f(e)\geq \frac{1}{2c}\cdot r$.\tcp*{$f(e)$ is large.} \label{pseudocode-fpt-select}
            
            Sample $X \in \{0,1\}$, where $\Pr[X=0]=\frac{q}{\lfloor 2c\cdot \log r\rfloor +q}$ and $\Pr[X=1]=\frac{\lfloor 2c\cdot \log r\rfloor}{\lfloor 2c\cdot \log r\rfloor +q}$.

            \If{$X=0$}{ \label{pseudocode-fpt-if1}
                \If{$f(\cN-e)<r$}{
                    $\cQ \gets \{\{e\}\}$.   \tcp*{$\{e\}$ is a quotient.} \label{pseudocode-fpt-singleton}
                }
                \Else{ \label{pseudocode-fpt-else2}
                    $\cQ \gets \Search(f_{\backslash e}, q-1)$.

                    Add element $e$ into every set in $\cQ$ if $\cQ$ is not FAIL.
                }

                \KwRet $\cQ$.
            }
            \Else{ \label{pseudocode-fpt-else1}

                \KwRet $\Search(f_e, q)$.
                
            }
        }
        \Else{ \label{pseudocode-fpt-else}
            $A\gets \emptyset$.
        
            \For{each $e\in \cN$}{
                Add element $e$ into $A$ with probability $\frac{1}{q}$.
            }

            \If{$f(A) < \frac{1}{8c}\cdot r$}{
                \KwRet FAIL.
            }

            \KwRet $\Search(f_A, q)$.
        }
    }
\end{algorithm2e}

The following lemma bounds the run-time and the number of quotients returned by Algorithm~\ref{algo:pseudocode-fpt}.

\begin{lemma}\label{lemma:fpt-runtime}
    Let $\cF$ be a $c$-\qbounded family of polymatroids and $f\in \cF$ with $q$ being the minimum size of a non-empty quotient of $f$. Algorithm~\ref{algo:pseudocode-fpt} on input $f$ and $q$ returns $n^{2c}$ quotients in $n^{O(c)}$ time, where $n$ denotes the size of the ground set of $f$.
\end{lemma}
\begin{proof}
    We first prove the upper bound on the number of quotients returned by induction on $n$. For the base case of $n\leq 2c$, the number of minimum quotients is at most $2^n\leq n^{2c}$ and Algorithm~\ref{algo:pseudocode-fpt} would use a brute-force search to find them in Step~\ref{pseudocode-fpt-bruteforce} since $q\leq n \leq 2c$.

    We now prove the induction step. If $r\leq 2c$, then by Lemma~\ref{lemma:quotient-counting}, there are at most $n^{2c}$ minimum non-empty quotients. Algorithm~\ref{algo:pseudocode-fpt} would use a brute-force search to find them in Step~\ref{pseudocode-fpt-bruteforce}. We now assume that $r>2c$. First, suppose there exists a heavy element $e\in \cN$. The algorithm executes the if-clause in Step~\ref{pseudocode-fpt-if} and branches to $X=0$ or $X=1$. For the case of $X=0$, the algorithm executes the if-clause in Step~\ref{pseudocode-fpt-if1}; it either returns $\{\{e\}\}$ or it recursively calls \Search$(f_{\backslash e}, q-1)$. By induction hypothesis, the number of quotients returned is at most
    $$\max\{1, (n-1)^{2c}\}<n^{2c}.$$
    For the case of $X=1$, the algorithm executes the else-clause in Step~\ref{pseudocode-fpt-else1}, where it recursively calls \Search$(f_e, q)$. By induction hypothesis, the number of quotients returned is at most
    $$(n-1)^{2c}<n^{2c}.$$

    Next, suppose that there is no heavy element. The algorithm executes the else-clause in Step~\ref{pseudocode-fpt-else}. We observe that the set $A$ of contracted elements has $f(A)\geq \frac{1}{8c}\cdot r$, which implies that $|A|\geq 1$. By induction hypothesis, the number of quotients returned is at most
    $$(n-1)^{2c}<n^{2c}.$$

    We now analyze the runtime. At each recursive call to $\Search(f, q)$, the size of the ground set decreases by at least $1$, so the recursion depth is at most $n$. At each level, the brute-force search in Line~\ref{pseudocode-fpt-bruteforce} is triggered only when $r\leq 2c$ or $q\leq 2c$, which takes at most $n^{2c}$ time. Thus, the overall runtime is $n^{O(c)}$.
\end{proof}

\subsection{Success Probability Analysis}\label{subsection:fpt-prob}

We now analyze the probability of a fixed quotient being in the collection returned by Algorithm~\ref{algo:pseudocode-fpt}. We first show the following lemma. Although the lemma is similar to that of the intermediate lemma shown by \cite{fomin2025fixed} for \HedgeMinCut, our proof of the lemma differs substantially from their proof since polymatroids do not have an underlying vertex-hedge structure. 

\begin{lemma}\label{lemma:fpt-random-sampling}
    Let $\cF$ be a $c$-\qbounded family of polymatroids and $f:2^{\cN}\rightarrow \mathbb{Z}_{\geq 0}$ be a polymatroid in $\cF$ of rank $r$. Let $Q$ be a minimum non-empty quotient of $f$ of size $q>2$ and $A\subseteq \cN$ be a random set, which includes every element $e\in \cN$ independently with probability $\frac{1}{q}$. Suppose $f(e)<\frac{r}{2c}$ for every element $e\in \cN$. Then,
    \begin{align}
        \Pr_A\left[A\cap Q=\emptyset \ \text{and} \ f(A)\geq \frac{r}{8c}\right] &\geq \frac{1}{32c}. \label{inequality:sampling-probability}
    \end{align}
\end{lemma}
\begin{proof}
    The probability that all elements in $Q$ are not contained in $A$ is
    $$\Pr_A[A\cap Q =\emptyset] = \left(1-\frac{1}{q}\right)^{q} \geq \frac{1}{4}.$$
    We now analyze the probability that $f(A)\geq \frac{r}{8c}$ conditioning on $A\cap Q=\emptyset$. Let $n:=|\cN|$ and $\sigma:=(e_1, e_2, \ldots, e_n)$ be a uniformly random permutation of elements in $\cN$. For every $i\in [0, n]$, we define $T_i:=\{e_j:j\in [1,i]\}$, $A_i:=A\cap T_i$, and $h_i:=r-f(A_i)$. We note that $A_0=T_0=\emptyset$ and $h_0=r$.

    Let $i\in [n]$. We consider the distribution of element $e_i$ conditioned on $A_{i-1}$ and $A\cap Q=\emptyset$. We claim that for every $e\in \cN\setminus (Q\cup A_{i-1})$, 
    $$\begin{aligned}
        \Pr_{\sigma, A}[e_i=e|A_{i-1} \ \text{and} \ A\cap Q=\emptyset]\geq \frac{1}{n-|A_{i-1}|}.
    \end{aligned}$$
    Indeed, for every element $x\in Q$, swapping the positions of $x$ and $e$ in the permutation, while leaving $A$ unchanged, gives a probability-preserving injection from conditioned outcomes with $e_i=x$ to conditioned outcomes with $e_i=e$. The conditioning is preserved because, if $e$ originally appeared after position $i$, then the first $i-1$ positions are unchanged by the swap, and hence $A_{i-1}$ is unchanged. If $e$ originally appeared before position $i$, then $e\notin A$, since otherwise $e\in A_{i-1}$. After the swap, $x$ occupies this earlier position, but $x\notin A$ because $A\cap Q = \emptyset$; hence $A_{i-1}$ is again unchanged. Moreover, $A\cap Q=\emptyset$ remains true because $A$ itself is unchanged. Therefore, $e$ is at least as likely to be $e_i$ as every element of $Q$. By symmetry, it is equally likely as every other element of $\cN \setminus (Q\cup A_{i-1})$. Thus, its conditional probability is at least the average over the $n-|A_{i-1}|$ possible elements of $e_i$.
    
    We note that if $e_i\in Q$, then it is not contained in $A_i$. Otherwise, $e_i\in \cN\setminus (Q\cup A_{i-1})$ and is contained in $A_i$ with probability $\frac{1}{q}$. Hence, we have
    \begin{align}
        &\E_{\sigma, A}[h_{i-1}-h_i|A_{i-1} \ \text{and} \ A\cap Q =\emptyset] \notag\\
        \geq& \sum_{e\in \cN\setminus (Q\cup A_{i-1})} \frac{1}{n-|A_{i-1}|}\cdot \frac{1}{q} \cdot (f(A_{i-1}+e)-f(A_{i-1})) \notag\\
        \geq& \frac{1}{n}\cdot \frac{1}{q}\cdot \sum_{e\in \cN\setminus (Q\cup A_{i-1})} (f(A_{i-1}+e)-f(A_{i-1})) \notag\\
        =& \frac{1}{n}\cdot \frac{1}{q}\cdot \left(\sum_{e\in \cN} (f(A_{i-1}+e)-f(A_{i-1})) - \sum_{e\in Q}(f(A_{i-1}+e) - f(A_{i-1}))\right) \notag\\
        \geq& \frac{1}{n}\cdot \frac{1}{q}\cdot \left(\sum_{e\in \cN} (f(A_{i-1}+e)-f(A_{i-1})) - q\cdot \frac{r}{2c}\right), \label{inequality:fpt-permutation-analysis}
    \end{align}
    where the last inequality holds since $f(A_{i-1}+e)-f(A_{i-1})\leq f(e)<\frac{r}{2c}$. By Lemma~\ref{lemma:quotient-property}, $Q$ is also the minimum non-empty quotient of $f_{A_{i-1}}$ since $Q\cap A_{i-1}=\emptyset$. We recall that $\cF$ is a $c$-\qbounded family of polymatroids with $f\in \cF$ and $f_{A_{i-1}}$ is a minor of $f$. Thus,
    $$q\leq c\cdot \frac{\sum_{e\in \cN\setminus A_{i-1}}f_{A_{i-1}}(e)}{f_{A_{i-1}}(\cN\setminus A_{i-1})}=c\cdot \frac{\sum_{e\in \cN}(f(A_{i-1}+e)-f(A_{i-1}))}{f(\cN)-f(A_{i-1})},$$
    which further implies that
    $$\sum_{e\in \cN}(f(A_{i-1}+e)-f(A_{i-1})) \geq \frac{q}{c}\cdot (f(\cN)-f(A_{i-1}))=\frac{q}{c}\cdot h_{i-1}.$$
    Therefore, plugging back into inequality~(\ref{inequality:fpt-permutation-analysis}), we have
    $$\begin{aligned}
        \E_{\sigma, A}[h_{i-1}-h_i|A_{i-1} \ \text{and} \ A\cap Q =\emptyset] \geq& \frac{1}{n}\cdot \frac{1}{q}\cdot \left(\sum_{e\in \cN} (f(A_{i-1}+e)-f(A_{i-1})) - q\cdot \frac{r}{2c}\right) \\
        \geq& \frac{1}{n}\cdot \frac{1}{q}\cdot \left(\frac{q}{c}\cdot h_{i-1} - q\cdot \frac{r}{2c}\right) \\
        =& \frac{1}{n}\cdot \frac{1}{c}\cdot \left(h_{i-1}-\frac{r}{2}\right).
    \end{aligned}$$

    This implies that
    $$\begin{aligned}
        \E_{\sigma,A}\left[h_i-\frac{r}{2}|A\cap Q =\emptyset\right] &= \E_{\sigma,A}\left[\left(h_{i-1}-\frac{r}{2}\right) - \E_{\sigma,A}\left[h_{i-1}-h_i|A_{i-1}\ \text{and}\ A\cap Q =\emptyset\right]|A\cap Q =\emptyset\right] \\
        &\leq \E_{\sigma,A}\left[\left(h_{i-1}-\frac{r}{2}\right) -\frac{1}{n}\cdot \frac{1}{c}\cdot \left(h_{i-1}-\frac{r}{2}\right)|A\cap Q =\emptyset\right] \\
        &= \left(1-\frac{1}{c\cdot n}\right) \cdot \E_{\sigma, A}\left[h_{i-1}-\frac{r}{2}|A\cap Q =\emptyset\right],
    \end{aligned}$$
    and moreover,
    $$\begin{aligned}
        \E_{\sigma,A}\left[h_n-\frac{r}{2}|A\cap Q =\emptyset\right] &\leq \left(1-\frac{1}{c\cdot n}\right) \cdot \E_{\sigma, A}\left[h_{n-1}-\frac{r}{2}|A\cap Q =\emptyset\right] \\
        &\leq \left(1-\frac{1}{c\cdot n}\right)^2 \cdot \E_{\sigma, A}\left[h_{n-2}-\frac{r}{2}|A\cap Q =\emptyset\right] \\
        &\leq \left(1-\frac{1}{c\cdot n}\right)^n \cdot \E_{\sigma, A}\left[h_{0}-\frac{r}{2}|A\cap Q =\emptyset\right] \\
        &= \left(1-\frac{1}{c\cdot n}\right)^n \cdot \frac{r}{2}.
    \end{aligned}$$
    This further shows that
    $$\begin{aligned}
        \E_{A}[f(\cN)-f(A)|A\cap Q =\emptyset] &= \frac{r}{2} + \E_{\sigma, A}\left[h_n-\frac{r}{2}|A\cap Q =\emptyset\right] \\
        &\leq \frac{r}{2}+\left(1-\frac{1}{c\cdot n}\right)^{n} \cdot \frac{r}{2} \\
        &\leq \frac{r}{2}+\exp\left(-\frac{1}{c}\right)\cdot \frac{r}{2} \\
        &\leq \left(1-\frac{1}{4c}\right)\cdot r. \ \ \text{(since $\exp(-\frac{1}{x})\leq 1-\frac{1}{2x}$ for $x\geq1$)}\\
    \end{aligned}$$
    Therefore,
    $$\E_{A}[f(A)|A\cap Q=\emptyset] \geq \frac{r}{4c}.$$
    We note that
    $$\begin{aligned}
        \E_A[f(A)|A\cap Q =\emptyset] &= \Pr_A\left[f(A)<\frac{r}{8c}|A\cap Q =\emptyset\right] \cdot \E_A\left[f(A)|f(A)<\frac{r}{8c} \ \text{and} \ A\cap Q =\emptyset\right] \\
        &+ \Pr_A\left[f(A)\geq\frac{r}{8c}|A\cap Q =\emptyset\right] \cdot \E_A\left[f(A)|f(A)\geq\frac{r}{8c} \ \text{and} \ A\cap Q =\emptyset\right] \\
        &\leq 1\cdot \frac{r}{8c} + \Pr_A\left[f(A)\geq\frac{r}{8c}|A\cap Q =\emptyset\right] \cdot r,
    \end{aligned}$$
    which implies that
    $$\Pr_A\left[f(A)\geq \frac{r}{8c}|A\cap Q =\emptyset\right]\geq \frac{1}{8c}$$
    and
    $$\begin{aligned}
        \Pr_A\left[A\cap Q=\emptyset \ \text{and} \ f(A)\geq \frac{r}{8c}\right] &= \Pr_A[A\cap Q=\emptyset] \cdot \Pr_A\left[f(A)\geq \frac{r}{8c}|A\cap Q =\emptyset\right] \\
        &\geq \frac{1}{4}\cdot \frac{1}{8c}=\frac{1}{32c}.
    \end{aligned}$$    
\end{proof}

We now analyze the probability of a fixed quotient being in the collection as follows.

\begin{lemma}\label{lemma:fpt-success-probability}
    Let $\cF$ be a $c$-\qbounded family of polymatroids and $f\in \cF$. For every fixed minimum non-empty quotient $Q$ of $f$ with size $q$, the collection of quotients returned by 
    Algorithm~\ref{algo:pseudocode-fpt} on input $(f, q)$ contains $Q$ with probability at least
    $$\frac{1}{r^{80c\log (c+1)}\cdot \binom{\lfloor 2c\cdot \log r\rfloor+q}{q}},$$
    where $r$ denotes the rank of $f$.
\end{lemma}
\begin{proof}
    For every positive integers $r$ and $q$, let $P(r,q)$ be the infimum over all $f\in \cF$ with rank $r$ with minimum non-empty quotient of size $q$ and all minimum non-empty quotient $Q$ of $f$, of the probability that the collection of quotients returned by 
    Algorithm~\ref{algo:pseudocode-fpt} on input $(f, q)$ contains $Q$. We now prove that
    \begin{align}\label{inequality:fpt-analysis-probability}
        P(r, q)\geq \frac{1}{r^{80c\log (c+1)} \cdot \binom{\lfloor 2c\cdot \log r\rfloor+q}{q}}
    \end{align}
    by induction on $r+q$. We also assume that $f(e)>0$ for every $e\in \cN$ since Algorithm~\ref{algo:pseudocode-fpt} will remove the other elements in Line~\ref{pseudocode-fpt-deleting}.
    
    For the case $r+q\leq 4c$, Algorithm~\ref{algo:pseudocode-fpt} uses a brute-force search in Step~\ref{pseudocode-fpt-bruteforce} to find minimum quotients of $f$ since $r\leq 2c$ or $q\leq 2c$. Consequently, the collection returned by Algorithm~\ref{algo:pseudocode-fpt} contains $Q$ with probability $1$. 

    We now prove the induction step when $r+q>4c$. We may assume that $r>2c$ and $q>2c$; otherwise, Algorithm~\ref{algo:pseudocode-fpt} would still use a brute-force search to find minimum quotients. Moreover, we recall that $\cF$ is a $c$-\qbounded family of polymatroids and $f\in \cF$. First suppose that there exists a heavy element. Let $e$ be such an element fixed in Step~\ref{pseudocode-fpt-select}. We consider the following two cases:
    \begin{enumerate}
        \item Suppose $e\in Q$: In this case, with probability $\frac{q}{\lfloor 2c\cdot \log r\rfloor+q}$, Algorithm~\ref{algo:pseudocode-fpt} branches to $X=0$ in Step~\ref{pseudocode-fpt-if1}. We observe that either $Q=\{e\}$ or $Q\supsetneq \{e\}$. If $Q=\{e\}$, then it follows that $f(\cN)>f(\cN-e)$ since $Q=\{e\}$ is a quotient and consequently, the algorithm will return $Q=\{e\}$ in Step~\ref{pseudocode-fpt-singleton}. If $Q\supsetneq \{e\}$, then $\{e\}$ is not a quotient and hence, the algorithm executes the else-clause in Step~\ref{pseudocode-fpt-else2}.
        It recursively calls \Search$(f_{\backslash e}, q-1)$ and adds element $e$ into all returned quotients. By Lemma~\ref{lemma:quotient-property}, $Q-e$ is the minimum non-empty quotient of $f_{\backslash e}$. This implies that the probability that $Q-e$ is contained in the collection of quotients returned by \Search$(f_{\backslash e}, q-1)$ is at least $P(r, q-1)$. Moreover, if $Q-e$ is contained in the collection returned, then $Q$ will be contained in the collection returned by the algorithm. Therefore,
        $$\begin{aligned}
            P(r, q) &\geq \frac{q}{\lfloor 2c\cdot \log r\rfloor+q}\cdot P(r, q-1) \\
            &\geq \frac{q}{\lfloor 2c\cdot \log r\rfloor+q}\cdot \frac{1}{r^{80c\log (c+1)} \cdot \binom{\lfloor 2c\cdot \log r\rfloor+q-1}{q-1}} \ \ \text{(by induction hypothesis)}\\
            &= \frac{1}{r^{80c\log (c+1)} \cdot \binom{\lfloor 2c\cdot \log r\rfloor+q}{q}}.
        \end{aligned}$$

        \item Suppose $e\not\in Q$: In this case, with probability $\frac{\lfloor 2c\cdot \log r\rfloor}{\lfloor 2c\cdot \log r\rfloor+q}$, Algorithm~\ref{algo:pseudocode-fpt} branches to $X=1$ in Step~\ref{pseudocode-fpt-else1}. The algorithm recursively calls \Search$(f_e, q)$. By Lemma~\ref{lemma:quotient-property}, $Q$ is still the minimum non-empty quotient of $f_e$. This implies that the probability that $Q$ is contained in the collection of quotients returned by \Search$(f_e, q)$ is at least $P(r-f(e), q)$. Moreover, if $Q$ is contained in the collection returned, then it will also be contained in the collection returned by the algorithm. Therefore,
        $$\begin{aligned}
            P(r, q) &\geq \frac{\lfloor 2c\cdot \log r\rfloor}{\lfloor 2c\cdot \log r\rfloor+q} \cdot P(r-f(e), q) \\
            &\geq \frac{\lfloor 2c\cdot \log r\rfloor}{\lfloor 2c\cdot \log r\rfloor+q} \cdot \frac{1}{(r-f(e))^{80c\log (c+1)} \cdot \binom{\lfloor 2c\cdot \log (r-f(e))\rfloor+q}{q}} \ \ \text{(by induction hypothesis)}\\
            &\geq \frac{\lfloor 2c\cdot \log r\rfloor}{\lfloor 2c\cdot \log r\rfloor+q} \cdot \frac{1}{(r-\frac{r}{2c})^{80c\log (c+1)} \cdot \binom{\lfloor 2c\cdot \log (r-\frac{r}{2c})\rfloor+q}{q}} \ \ \text{(since $f(e)\geq \frac{r}{2c}$)}\\
            &\geq \frac{1}{r^{80c\log (c+1)}} \cdot \frac{\lfloor 2c\cdot \log r\rfloor}{\lfloor 2c\cdot \log r\rfloor+q} \cdot \frac{1}{\binom{\lfloor 2c\cdot \log r + 2c\cdot \log(1-\frac{1}{2c})\rfloor+q}{q}} \\
            &\geq \frac{1}{r^{80c\log (c+1)}} \cdot \frac{\lfloor 2c\cdot \log r\rfloor}{\lfloor 2c\cdot \log r\rfloor+q} \cdot \frac{1}{\binom{\lfloor 2c\cdot \log r\rfloor-1+q}{q}} \ \ \text{(since $\log (1-\frac{1}{2c})\leq -\frac{1}{2c}$)}\\
            &= \frac{1}{r^{80c\log (c+1)}} \cdot \frac{1}{\binom{\lfloor 2c\cdot \log r\rfloor+q}{q}}.
        \end{aligned}$$
    \end{enumerate}
    Next, suppose every element $e\in \cN$ is light. Then, Algorithm~\ref{algo:pseudocode-fpt} will execute the else-clause in Step~\ref{pseudocode-fpt-else}. It constructs a random set $A$, which includes every element $e\in \cN$ with probability $\frac{1}{q}$ independently, and recursively calls \Search$(f_A, q)$. By Lemma~\ref{lemma:fpt-random-sampling}, the probability that none of the elements in $Q$ is contracted and $f(A)\geq \frac{r}{8c}$ is at least
        \begin{align}
            \Pr_A\left[A\cap Q=\emptyset \ \text{and} \ f(A)\geq \frac{r}{8c}\right] &\geq \frac{1}{32c}. \label{inequality:fpt-sampling-probability}
        \end{align}
        By Lemma~\ref{lemma:quotient-property}, $Q$ is still the minimum non-empty quotient of $f_A$ if $A\cap Q=\emptyset$. This implies that the probability that $Q$ is contained in the collection of quotients returned by \Search$(f_A, q)$ is at least $P(r-f(A), q)$. Moreover, if $Q$ is contained in the collection returned, then it will also be contained in the collection returned by the algorithm. Thus,
        $$\begin{aligned}
            P(r, q)&\geq \Pr_A\left[A\cap Q=\emptyset \ \text{and} \ f(A)\geq \frac{r}{8c}\right]\cdot P\left(r-\frac{1}{8c}\cdot r, q\right) \\
            &\geq \frac{1}{32c}\cdot P\left(r-\frac{1}{8c}\cdot r, q\right) \ \ \text{(by inequality~(\ref{inequality:fpt-sampling-probability}))}\\
            &\geq \frac{1}{32c}\cdot \frac{1}{((1-\frac{1}{8c})\cdot r)^{80c\log (c+1)} \cdot \binom{\lfloor 2c\cdot \log ((1-\frac{1}{8c})\cdot r)\rfloor+q}{q}} \ \ \text{(by induction hypothesis)}\\
            &\geq \frac{1}{32c}\cdot \frac{1}{(1-\frac{1}{8c})^{80c\log (c+1)}} \cdot \frac{1}{r^{80c\log (c+1)} \cdot \binom{\lfloor 2c\cdot \log r\rfloor+q}{q}} \\
            &\geq \frac{1}{r^{80c\log (c+1)} \cdot \binom{\lfloor 2c\cdot \log r\rfloor+q}{q}},
        \end{aligned}$$
        where the last inequality is because
        $$(1-\frac{1}{8c})^{80c\log (c+1)}\leq \exp(-10\log (c+1))\leq \frac{1}{32c}.$$
        
\end{proof}

%% file: strong-uniform-random-contraction-bounded-marginal.tex
\section{\MinQuo in $c$-\sqbounded $k$-polymatroids}\label{section:uniform-contraction-bounded-marginal}

In this section, we consider $c$-\sqbounded $k$-polymatroid families. We design a random contraction algorithm to find a subset $\cN'$ of the ground set such that the polymatroid obtained by contracting the complement of $\cN'$ has rank at most $\alpha(c+k)$ and moreover, every $\alpha$-approximate minimum quotient $Q$ of $f$ is contained in $\cN'$ with probability $\Omega(k^{-1} r^{-(c+1)\alpha})$. Our algorithm generalizes the random contraction for $k$-hypergraphs by Kogan and Krauthgamer \cite{kogan2015sketching}.
The following is the main theorem of this section.

\begin{theorem}\label{thm:uniform-random-contraction-strong}
    For an integer $k\ge 1$, let $\cF$ be a $c$-\sqbounded $k$-polymatroid family and $\alpha \geq 1$. 
    There exists a randomized algorithm that takes a polymatroid $f\in \cF$ as input (via its evaluation oracle) and runs in polynomial time to return a subset $\cN'$ of the ground set $\cN$ of $f$ such that 
    \begin{enumerate}
        \item the polymatroid obtained from $f$ by contracting $\cN\setminus \cN'$ has rank at most $\alpha (c+k)$ and 
        \item  every $\alpha$-approximate minimum quotient $Q$ of $f$ is contained in $\cN'$ with probability at least 
    $$\frac{(c+1)\alpha+1}{(c+k)\alpha+1}\cdot \binom{r-\alpha(k-1)}{(c+1)\alpha}^{-1}.$$
    \end{enumerate}
\end{theorem}

We first describe the uniform random contraction algorithm. Let $f:2^{\cN}\rightarrow \mathbb{Z}_{\geq 0}$ be the given polymatroid from a $c$-\sqbounded $k$-polymatroid family. Let $\alpha\geq 1$. If $f(\cN)\leq \alpha\cdot (c+k)$, then the algorithm returns the ground set $\cN$. Otherwise, among all elements $e\in \cN$ with $f(e)>0$, the algorithm samples one such element uniformly at random and recurses on the contracted polymatroid $f_e$. The pseudocode is given in Algorithm~\ref{algo:pseudocode-uniform-bounded-marginal-strong}.

\begin{algorithm2e}[H]
\caption{Uniform Random Contraction Algorithm}
\label{algo:pseudocode-uniform-bounded-marginal-strong}
\SetKwInput{KwInput}{Input}                
\SetKwInput{KwOutput}{Output}              
\LinesNumbered

\DontPrintSemicolon
  
    \KwInput{evaluation oracle access to polymatroid $f:2^{\cN}\rightarrow \Z_{\geq 0}$.
    }
    \KwOutput{a subset of the ground set $\cN$.}
    \SetKwFunction{Contract}{Contract}
    \SetKwProg{Fn}{Function}{:}{}
    
    \Fn{\Contract{$f$}}{

        \If{$f(\cN)\leq \alpha\cdot (c+k)$}{
            \KwRet $\cN$.
        }

        $S\gets \{e\in \cN: f(e)>0\}$. \label{pseudocode-kk-delete-strong}

        Sample element $e\in S$ uniformly at random.

        \KwRet $\Contract(f_e)$.
    }
\end{algorithm2e}

Let $\cN'\subseteq \cN$ be the subset returned by Algorithm~\ref{algo:pseudocode-uniform-bounded-marginal-strong}. We observe that the polymatroid obtained from $f$ by contracting $\cN\setminus \cN'$ has rank at most $\alpha\cdot (c+k)$.
Moreover, Algorithm~\ref{algo:pseudocode-uniform-bounded-marginal-strong} takes polynomial time since at each recursive call to \Contract$(f)$, the size of the ground set decreases by at least $1$. We now analyze the probability that a fixed $\alpha$-approximate minimum quotient is contained in the returned subset. For positive integers $i,k,\alpha\geq 1$, we define
\[
H(i,k,\alpha):=
\begin{cases}
    \frac{(c+1)\alpha+1}{(c+k)\alpha+1}\cdot \binom{i-\alpha(k-1)}{(c+1)\alpha}^{-1} & \text{ if }i>\alpha(c+k), \\
    1 & otherwise.
\end{cases}
\]
The following lemma will be used later. We defer the proof to Appendix~\ref{appendix:lp-analysis}.

\begin{restatable}{lemma}{lpanalysisstrong}\label{lemma:lp-analysis-1}
    Let $r, k, c, \alpha\geq 1$ be positive integers with $r>\alpha\cdot (c+k)$. Let $P:\mathbb{N}\rightarrow \mathbb{R}_{+}$ be a positive-valued function defined over the natural numbers. Then, the optimum value of the linear program~(\ref{lp-1}) defined below is at least $\min_{i=1}^{k} \frac{r-\alpha\cdot (c+k)}{r+\alpha\cdot (i-k)} \cdot P(r-i)$.
    \begin{equation}
    \label{lp-1}
    \begin{aligned}
        \min\quad &\sum_{i=1}^{k}(x_i-y_i)\cdot P(r-i) \\
        s.t.\quad & 0\leq y_i\leq x_i \quad \forall\ i \in [1, k] \\
            &\sum_{i=1}^{k}x_i=1\\
            &\sum_{i=1}^{k}y_i\leq \frac{\alpha}{r}\cdot \sum_{i=1}^{k}(i+c)\cdot x_i
    \end{aligned}
    \end{equation}
\end{restatable}

Lemma~\ref{lemma:uniform-random-contraction-strong} that is shown below implies Theorem~\ref{thm:uniform-random-contraction-strong}.

\begin{lemma}\label{lemma:uniform-random-contraction-strong}
    Let $\cF$ be a $c$-\sqbounded $k$-polymatroid family with $\alpha,k \geq 1$. Let $f:2^{\cN}\rightarrow \mathbb{Z}_{\geq 0}$ be a polymatroid in $\cF$ with rank $r>\alpha(c+k)$ and $q$ be the minimum size of non-empty quotients of $f$. Then, every $\alpha$-approximate minimum quotient $Q$ of $f$ is contained in the subset $\cN'$ returned by Algorithm~\ref{algo:pseudocode-uniform-bounded-marginal-strong} on input $f$ with probability at least $H(r,k,\alpha)$.
\end{lemma}
\begin{proof}
    For every positive integer $r$, let $P(r)$ be the infimum over all $f\in \cF$ with rank $r$ and $f(e)\leq k$ for every element $e$ in the ground set of $f$ and all $\alpha$-approximate minimum quotient $Q$ of $f$, of the probability that $Q$ is contained in the subset $\cN'$ returned by $\Contract(f)$. We now prove that $$P(r)\geq H(r,k,\alpha)$$ by induction or $r$. We recall that $f(e)\leq k$ for every $e\in \cN$.
    
    For the base case, we consider $r\leq \alpha\cdot (c+k)$. Algorithm~\ref{algo:pseudocode-uniform-bounded-marginal-strong} returns the whole ground set $\cN$, which implies that $P(r)=1\geq H(r,k,\alpha)$.
    Next, we prove the induction step. Let $q$ be the minimum size of non-empty quotients of $f$ and $|Q|\leq \alpha\cdot q$. Let $S$ be the set obtained in Line~\ref{pseudocode-kk-delete-strong}, where $f(e)\geq 1$ for every $e\in S$. We observe that $Q\subseteq S$, since every element $e\in \cN$ with $f(e)=0$ cannot be in the quotient. For every integer $1\leq i \leq k$, we define
    $$x_i:=\frac{|\{e\in S: f(e)=i\}|}{|S|} \text{ \ and \ } y_i:=\frac{|\{e\in Q: f(e)=i\}|}{|S|},$$
    which further shows that $0\leq y_i\leq x_i$. Moreover, since $f\in \cF$, which is a $c$-\sqbounded polymatroid family, we have 
    $$\begin{aligned}
        q \leq \sum_{e\in S}\frac{f(e)+c}{r}.
    \end{aligned}$$
    This further shows that
    $$\begin{aligned}
        \sum_{i=1}^{k}y_i & = \frac{|Q|}{|S|} \leq \frac{\alpha\cdot q}{|S|} \leq \frac{\alpha}{r} \cdot \sum_{e\in S}\frac{f(e)+c}{|S|} = \frac{\alpha}{r} \cdot \sum_{i=1}^{k}(i+c)\cdot x_i.
    \end{aligned}$$
    We observe that Algorithm~\ref{algo:pseudocode-uniform-bounded-marginal-strong} contracts an element $e\in S$, which is sampled uniform at random, and recurses on the contracted polymatroid $f_e$. By Lemma~\ref{lemma:quotient-property}, the minimum non-empty quotient of $f_e$ still has size at least $q$.
    Thus, if $e\not\in Q$, $Q$ is still a $\alpha$-approximate minimum quotient in $f_e$, and moreover, the rank of $f_e$ is $r-f(e)$ and hence, the set $Q$ is contained in the subset returned by the algorithm executed on $f_e$ with probability at least $P(r-f(e))$. The quotient $Q$ is contained in the subset returned by the algorithm if (i) the sampled element $e$ is not from $Q$ and (ii) $Q$ is contained in the subset returned by the algorithm when executed on the contracted polymatroid $f_e$. This implies that
    $$\begin{aligned}
        P(r)&\geq \sum_{i=1}^{k}\Pr[f(e)=i \text{ \ and \ } e\notin Q ] \cdot P(r-i) \\
        &= \sum_{i=1}^{k}(x_i-y_i)\cdot P(r-i).
    \end{aligned}$$
    Therefore, $P(r)$ is at least the optimum value of the following linear program:
    \begin{equation}
    \label{kogan-lp-strong}
    \begin{aligned}
        \min\quad &\sum_{i=1}^{k}(x_i-y_i)\cdot P(r-i) \\
        s.t.\quad & 0\leq y_i\leq x_i \quad \forall\ i \in [1, k] \\
            &\sum_{i=1}^{k}x_i=1\\
            &\sum_{i=1}^{k}y_i\leq \frac{\alpha}{r}\cdot \sum_{i=1}^{k}(i+c)\cdot x_i
    \end{aligned}
    \end{equation}
    
    By Lemma~\ref{lemma:lp-analysis-1}, there exists an index $1\leq i \leq k$ with
    $$P(r) \geq \frac{r-\alpha\cdot (c+k)}{r+\alpha\cdot (i-k)} \cdot P(r-i).$$
    
    If $r-i> \alpha\cdot (c+k)$, then by defining $r':=r-\alpha\cdot (c+k)$, we have
    $$\begin{aligned}
        \frac{H(r-i, k, \alpha)}{H(r, k, \alpha)} &= \frac{\binom{r'+(c+1)\alpha}{(c+1)\alpha}}{\binom{r'-i+(c+1)\alpha}{(c+1)\alpha}} = \frac{(r'+(c+1)\alpha)\ldots(r'-i+(c+1)\alpha+1)}{r'\ldots (r'-i+1)} \\
        &\geq \left(1+\frac{(c+1)\alpha}{r'}\right)^i \geq 1+\frac{(c+1)\alpha}{r'}\cdot i\\
        &= \frac{r+\alpha\cdot (i-k)+\alpha\cdot c\cdot (i-1)}{r-\alpha\cdot (c+k)} \\
        &\geq \frac{r+\alpha\cdot (i-k)}{r-\alpha\cdot (c+k)},
    \end{aligned}$$
    which further shows that
    $$\begin{aligned}
        P(r) &\geq \frac{r-\alpha\cdot (c+k)}{r+\alpha\cdot (i-k)} \cdot P(r-i) \\
        &\geq \frac{r-\alpha\cdot (c+k)}{r+\alpha\cdot (i-k)} \cdot H(r-i, k, \alpha)\ \ \text{(by induction hypothesis)}\\
        &\geq H(r, k, \alpha).
    \end{aligned}$$

    If $r-i\leq\alpha \cdot (c+k)$, then $P(r-i)=1$ and we have
    $$\begin{aligned}
        P(r) &\geq \frac{r-\alpha\cdot (c+k)}{r+\alpha\cdot (i-k)} \cdot P(r-i) = \frac{r-\alpha\cdot (c+k)}{r+\alpha\cdot (i-k)}\\
        &= 1 - \frac{\alpha\cdot (i+c)}{r+\alpha\cdot (i-k)} \geq 1 - \frac{\alpha\cdot (i+c)}{1+\alpha\cdot (i+c)} \ \ \text{(since $r\geq \alpha \cdot (c+k) +1$)} \\
        &\geq \frac{1}{\alpha\cdot (k+c) +1} \geq \frac{1}{\alpha\cdot (k+c) +1}\cdot \frac{(c+1)\alpha+1}{\binom{r-\alpha(k-1)}{(c+1)\alpha}} = H(r, k, \alpha),
    \end{aligned}$$
    where the last inequality is because $r\geq \alpha \cdot (c+k) +1$ and $\binom{r-\alpha(k-1)}{(c+1)\alpha}\geq \binom{(c+1)\alpha+1}{(c+1)\alpha}=(c+1)\alpha+1$.

\end{proof}

For $c$-\qbounded $k$-polymatroid families, we obtain similar results; see Appendix~\ref{section:uniform-contraction-bounded-marginal-appendix} for details.

%% file: quotient-bounded-instance.tex
\section{Graphs, Hypergraphs, and Hedgegraphs}\label{section:quotient-bounded-instance}
In this section, we discuss graphic matroids, hypergraphic polymatroids, and hedgegraph polymatroids. We first show that if a hedgegraph has span at most $s$, then the corresponding hedgegraph polymatroid is $s$-\sqbounded. We formalize it as follows.

\begin{lemma}\label{lemma:hedgegraph-sqbounded}
    For a hedgegraph $G$, let $f_G$ denote the hedgegraph polymatroid associated with $G$. Let $\mathcal{G}_s$ denote the collection of hedgegraphs of span at most $s$ and let $\mathcal{F}_s:=\{f_G: G\in \mathcal{G}_s\}$. Then, the family $\mathcal{F}_s$ is minor-closed and moreover, $\mathcal{F}_s$ is $s$-\sqbounded.
\end{lemma}
\begin{proof}
    Let $G=(V, E)$ and $f_G$ be the hedgegraph polymatroid associated with $G$. Then, $f_{\backslash e}$ and $f_e$ are polymatroids associated with the hedgegraph $G_{\backslash e}$ obtained by deleting $e$ from $G$ and the hedgegraph $G_e$ obtained by contracting $e$ in $G$ respectively. Moreover, $G_{\backslash e}$ and $G_e$ have span at most $s$. This implies that $\cF_s$ is minor-closed. We now prove that $\cF_s$ is $s$-\sqbounded.

    Let $G=(V,E) \in \mathcal{G}_s$. If $f_G(E)=0$, then $f_G$ is the empty polymatroid and it trivially satisfies the $s$-\sqbounded property. Henceforth, we assume that $f_G(E)\ge 1$. Let $I\subsetneq V$ be the isolated vertices in $G$. We recall that for every $A\subseteq E$, $f_G(A)=|V|-\Comp(V,A)$, where $\Comp(V,A)$ is the number of connected components in the hypergraph $(V,\{h: h\in e\text{ for some }e\in A\})$. Thus, $f_G(E)=|V|-\Comp(V,E)\leq |V\setminus I|$.
    Moreover, for every hedge $e\in E$, we have $f_G(e) \geq |e|-s$, where $|e|$ is the number of vertices contained in $e$.
    Therefore, we have
    $$\frac{\sum_{e\in E}(f_G(e)+s)}{f_G(E)}\geq \frac{\sum_{e\in E}|e|}{f_G(E)}\geq \frac{\sum_{e\in E}|e|}{|V\setminus I|}.$$
    We observe that the ratio on the right hand side equals the average degree of vertices in $V\setminus I$, which is at least the size of minimum non-empty minimum quotient of $f_G$, since there is a vertex $v\in V\setminus I$ with degree at most the average degree and the hedges incident to $v$ form a valid quotient. Consequently, $\cF_s$ is $s$-\sqbounded.
\end{proof}

For a graph $G$, let $r_G$ denote the graphic matroid associated with $G$. Let $\mathcal{G}_{\text{graph}}$ denote the collection of graphs and $\cF_{\text{graph}}:=\{r_G:G\in \mathcal{G}_{\text{graph}}\}$. For a hypergraph $G$, let $f_G$ denote the hypergraphic polymatroid associated with $G$. Let $\mathcal{G}_{\text{hypergraph}}$ denote the collection of hypergraphs and $\cF_{\text{hypergraph}}:=\{f_G:G\in\mathcal{G}_{\text{hypergraph}}\}$. We observe that graphs and hypergraphs are also hedgegraphs with span at most $1$. Hence, $\cF_{\text{graph}}$ and $\cF_{\text{hypergraph}}$ are $1$-\sqbounded.

\begin{corollary}\label{lemma:hypergraph-1-strong-bounded}
    $\cF_{\text{graph}}$ and $\cF_{\text{hypergraph}}$ are $1$-\sqbounded.
\end{corollary}

We now prove that hedgegraph polymatroid family is $2$-\qbounded.

\begin{lemma}\label{lemma:hedgegraph-2qbounded}
    For a hedgegraph $G$, let $f_G$ denote the hedgegraph polymatroid associated with $G$. Let $\mathcal{G}_{\text{hedgegraph}}$ denote the family of hedgegraphs and let $\mathcal{F}_{\text{hedgegraph}}:=\{f_G: G\in \mathcal{G}_{\text{hedgegraph}}\}$. Then, the family $\mathcal{F}_{\text{hedgegraph}}$ is minor closed and moreover, $\mathcal{F}_{\text{hedgegraph}}$ is $2$-\qbounded.
\end{lemma}
\begin{proof}
        Let $G=(V, E)$ and $f_G$ be the hedgegraph polymatroid associated with $G$. Then, $f_{\backslash e}$ and $f_e$ are polymatroids associated with the hedgegraph $G_{\backslash e}$ obtained by deleting $e$ from $G$ and the hedgegraph $G_e$ obtained by contracting $e$ in $G$ respectively. This implies that $\cF_{\text{hedgegraph}}$ is minor-closed. We now prove that $\cF_{\text{hedgegraph}}$ is $2$-\sqbounded.
 
    Let $G=(V,E)\in \mathcal{G}$. If $f_G(E)=0$, then $f_G$ is the empty polymatroid and it trivially satisfies the $2$-\qbounded property. Let $I\subsetneq V$ be the isolated vertices in $G$. We recall that every $A\subseteq E$, $f_G(A)=|V|-\Comp(V,A)$. Thus, $f_G(E)=|V|-\Comp(V,E)\leq |V\setminus I|$.
    Moreover, for every hedge $e\in E$, we have $f_G(e) \geq |e|-\text{span}(e)\geq |e|/2$, where $|e|$ is the number of vertices contained in $e$ and $\text{span}(e)$ is the number of connected components of $e$.
    Therefore, we have
    $$2\cdot \frac{\sum_{e\in E}f_G(e)}{f_G(E)}\geq  \frac{\sum_{e\in E}|e|}{f_G(E)}\geq \frac{\sum_{e\in E}|e|}{|V\setminus I|}.$$
    We observe that the ratio on the right hand side equals the average degree of vertices in $V\setminus I$, which is at least the size of non-empty minimum quotients of $f_G$, since there is a vertex $v\in V\setminus I$ with degree at most the average degree and the hedges incident to $v$ form a valid quotient. This further implies that $f_G$ is $2$-\qbounded.
\end{proof}

Since hedgegraphs capture graphs and hypergraphs, Lemma~\ref{lemma:hedgegraph-2qbounded} implies that $\cF_{\text{graph}}$ and $\cF_{\text{hypergraph}}$ are $2$-\qbounded.

\begin{corollary}
    $\cF_{\text{graph}}$ and $\cF_{\text{hypergraph}}$ are $2$-\qbounded.
\end{corollary}

We recall that every hyperedge in a $k$-hypergraph contains at most $k$ vertices. We now recover the known bound on the number of $\alpha$-approximate minimum cuts in a $k$-hypergraph using Theorem \ref{thm:uniform-random-contraction-strong}. For a connected hypergraph $G=(V, E)$, a subset $\emptyset\neq S\subsetneq V$ is an $\alpha$-approximate min-cut in $G$ if $|\delta_G(S)|\le \alpha \min\{|\delta_G(T)|: \emptyset\neq T\subsetneq V\}$. 
\begin{corollary}
    For every connected $k$-hypergraph $G$, the number of $\alpha$-approximate min-cuts in $G$ is $O(2^{\alpha k}\cdot k\cdot n^{2\alpha})$, where $n$ denotes the number of vertices in $G$.
\end{corollary}
\begin{proof}
    Let $\mathcal{G}_k$ denote the collection of $k$-hypergraphs. We define $\cF_k:=\{f_G:G\in \mathcal{G}_k\}$, where $f_G$ denotes the hypergraphic polymatroid associated with $G$. We observe that $\cF_k$ is a family of $(k-1)$-polymatroids.
    For every hypergraph $G=(V, E)$ and $e\in E$, the polymatroids $f_{\backslash e}$ and $f_e$ correspond to the hypergraphic polymatroid associated with $G-e$ and $G/e$ respectively; moreover, if $G$ is a $k$-hypergraph, then both $G-e$ and $G/e$ are both $k$-hypergraphs.
    This implies that $\cF_k$ is minor-closed. By Corollary~\ref{lemma:hypergraph-1-strong-bounded}, $\cF_{\text{hypergraph}}$ is minor-closed and $1$-\sqbounded. Since $\mathcal{G}_k\subseteq \mathcal{G}_{\text{hypergraph}}$, $\cF_k$ is also $1$-\sqbounded.

    Now, we fix a connected $k$-hypergraph $G=(V,E)$. According to Theorem~\ref{thm:uniform-random-contraction-strong}, there exists a randomized algorithm that returns a set $F$ of hyperedges such that
    \begin{enumerate}
        \item the hypergraph $H$ obtained from $G$ by contracting hyperedges $E\setminus F$ has at most $\alpha k+1$ vertices and
        \item every $\alpha$-approximate min-cut $C$ of $G$ is contained in $F$ with probability at least
        $$\Pr[C\subseteq F]\geq \frac{2\alpha+1}{\alpha k+1}\cdot \binom{n-1-\alpha(k-2)}{2\alpha}^{-1}.$$
    \end{enumerate}
    Since $H$ has at most $\alpha k+1$ vertices, the number of cuts in hypergraph $H$ is at most $2^{\alpha k}$. This implies that the number of $\alpha$-approximate min-cuts in hypergraph $G$ is at most
    $$\begin{aligned}
        2^{\alpha k}\cdot \frac{\alpha k+1}{2\alpha+1}\cdot \binom{n-1-\alpha(k-2)}{2\alpha} =O(2^{\alpha k}\cdot k\cdot n^{2\alpha}).
    \end{aligned}$$
\end{proof}


%% file: conclusion.tex
\section{Conclusion}
We generalized Karger's influential random contraction technique for min-cut in graphs to a framework to find min-quotient in polymatroids. Min-quotient in arbitrary matroids is NP-hard. We introduced the notion of \qbounded polymatroids as a subfamily of polymatroids that admit efficient algorithms for min-quotient via the random contraction technique. Our work raises two natural directions for further research. Firstly, for $c$-\sqbounded polymatroid families, we have given a randomized polynomial-time algorithm to find a minimum non-empty quotient. Is it possible to design a deterministic polynomial-time algorithm for the same? This would lead to a deterministic polynomial-time algorithm to find a minimum cut in constant span hedgegraphs. Secondly, in addition to graphs, hypergraphs, and hedgegraphs, are there other interesting families of $c$-\qbounded and $c$-\sqbounded polymatroids?

%% file: quotient-property.tex
\section{Missing Proofs from Section~\ref{section:preliminaries}}\label{appendix:quotient}

In this section, we restate and prove Lemmas~\ref{lemma:quotient-property} and \ref{lemma:quotient-counting}.

\quotientproperty*
\begin{proof}
    We prove the four properties separately as follows.

    \begin{enumerate}
        \item Let $S:=\cN \setminus Q$, which is a closed set of $f$. For every element $e\in Q$, $f(S+e)-f(S)>0$. We define $S':=S\setminus T$ and $Q':=Q\setminus T$. For every element $e\in Q'$, we have
        $$\begin{aligned}
            f_{\backslash T}(S'+e)-f_{\backslash T}(S')&=f(S'+e)-f(S')\\
            &\geq f(S+e) - f(S) \ \ \text{(by submodularity of $f$)}\\
            &>0,
        \end{aligned}$$
        which further shows that $Q'$ is a quotient of $f_{\backslash T}$.
    
        \item We observe that $Q-e$ is non-empty since $|Q|>1$. By the previous property, $Q - e$ is a quotient of $f_{\backslash e}$. We now show that  $Q - e$ is a minimum non-empty quotient of $f_{\backslash e}$. Suppose there is a non-empty quotient $Q'$ of $f_{\backslash e}$ with $|Q'|<|Q|-1$. Let $S':=(\cN-e)\setminus Q'$. Thus, for every element $e'\in Q'$,
        $$\begin{aligned}
            f(S'+e')-f(S') &= f_{\backslash e}(S'+e') - f_{\backslash e}(S') > 0,
        \end{aligned}$$
        where the last inequality is because $S'$ is a closed set of $f_{\backslash e}$. This shows that in the polymatroid $f$, $\Span(S')$ is either $S'$ itself or $S'+e$, which implies that $Q'$ or $Q'+e$ is a quotient of $f$. This contradicts the minimality of $Q$.

        \item We first show that $Q$ is a quotient of $f_e$. Let $S:=(\cN - e)\setminus Q$. For every element $e'\in Q$,
        $$\begin{aligned}
            f_e(S+e')-f_e(S) &= (f(S+e'+e)-f(e)) - (f(S+e)-f(e)) \\
            &= f(S+e'+e)-f(S+e) \\
            &>0,
        \end{aligned}$$
        where the last inequality is because $S+e=\cN\setminus Q$ is a closed set of $f$. This shows that $Q$ is a quotient of $f_e$.

        Suppose there is a non-empty quotient $Q'$ of $f_e$ with $|Q'|<|Q|$. Let $S':=(\cN-e)\setminus Q'$. Thus, for every element $e'\in Q'$, we have 
        $$\begin{aligned}
            f(S'+e+e')-f(S'+e) &= (f(S'+e+e')-f(e)) - (f(S'+e)-f(e)) \\
            &= f_e(S'+e') - f_e(S') \\
            &> 0,
        \end{aligned}$$
        where the last inequality is because $S'$ is a closed set of $f_e$. This shows that $S'+e$ is a closed set of $f$ and $Q'$ is a quotient of $f$. This contradicts the minimality of $Q$.

        \item Let $q$ be the minimum size of a non-empty quotient of $f$, $e\in \cN$, $Q$ be a non-empty quotient of $f_e$, and $S:=(\cN-e)\setminus Q$. For every element $e'\in Q$, we have 
        $$\begin{aligned}
            f(S+e+e')-f(S+e) &= (f_e(S+e')+f(e)) - (f_e(S)+f(e)) \\
            &= f_e(S+e')-f_e(S) \\
            &> 0,
        \end{aligned}$$
        where the last inequality is because $Q$ is a quotient of $f_e$. This shows that $S+e$ is a closed set of $f$. Hence, $Q$ is a non-empty quotient of $f$ and consequently, $|Q|\geq q$.
    \end{enumerate}
\end{proof}

Next, we show the upper bound on the number of non-empty quotients of polymatroids in Lemma~\ref{lemma:quotient-counting}.
\lemmacounting*
\begin{proof}
    For the case of $n=1$, there is only one non-empty quotient. We may now assume that $n\geq 2$. We first show that every non-empty quotient $Q$ of $f$ can be written as $Q=\cN\setminus \text{span}(R)$ for some $R\subseteq \cN$ with $|R|<r$. Let $S:=\cN\setminus Q$, which is a closed set. We observe that $f(S)<r$, otherwise $S=\cN$ and $Q$ is empty. Starting with $R=\emptyset$, we repeatedly add an element from $S\setminus \text{span}(R)$ to $R$ as long as such an element exists. Each added element strictly increases $f(R)$ by at least $1$, so the process terminates after at most $f(S)<r$ steps with $\text{span}(R)=S$ and equivalently $Q=\cN\setminus \text{span}(R)$. Hence, every non-empty quotient is induced by some $R\subseteq \cN$ with $|R|<r$.

    Therefore, the number of non-empty quotients is at most the number of subsets of $\cN$ of size at most $r-1$, which is
    $$\sum_{i=0}^{r-1}\binom{n}{i}\leq n^r,$$
    where the last inequality holds since $n\geq 2$.
\end{proof}

%% file: gkp-counting-bound.tex
\section{Proof of Corollary~\ref{coro:qbounded-number-1}}\label{appendix:gkp-counting-bound}
In this section, we consider $c$-\qbounded families of polymatroids and analyze the number of minimum non-empty quotients. We analyze the probability of a fixed minimum non-empty quotient being in the collection returned by Algorithm~\ref{algo:pseudocode-gkp}. Lemmas~\ref{lemma:gkp-runtime} and \ref{lemma:gkp-counting-success-probability} together imply Corollary~\ref{coro:qbounded-number-1}.

\begin{lemma}\label{lemma:gkp-counting-success-probability}
  Let $\cF$ be a $c$-\qbounded family of polymatroids and $f\in \cF$ with $r$ being the rank of $f$. Let $Q$ be an arbitrary minimum non-empty quotient of $f$, and let $q = |Q|$. Then, the probability that $Q$ is contained in the collection of quotients returned by Algorithm~\ref{algo:pseudocode-gkp} on input $f$ is at least $1/r^{O(c\log{q})}$.
\end{lemma}

\begin{proof}
  For every integer $r\geq 1$ and $\ell\geq 0$, let $P(r, \ell)$ be the infimum over all $f\in \cF$ with rank $r$ on which \text{\Search}$(f)$ performs at most $\ell$ branching steps and all minimum non-empty quotient $Q$ of $f$, of the probability that the collection of quotients returned by \text{\Search}$(f)$ contains quotient $Q$. By Lemma~\ref{lemma:gkp-depth}, we may assume that $\ell \leq 4c\cdot \log r$. We now prove that
    \begin{align}\label{inequality:gkp-analysis-probability-appendix}
        P(r, \ell)\geq \left(\frac{1}{r}\right)^{4c} \cdot \left(\frac{1}{2(q+1)}\right)^{\ell}    
    \end{align}
    by induction on $r+\ell$. We may also assume that every element $e\in \cN$ has $f(e)\geq 1$ since all elements $e\in \cN$ with $f(e)=0$ are not in $Q$ by definition of $Q$ being a quotient and such elements are removed by the algorithm in Step~\ref{pseudocode-gkp-deleting}.
    
    For the base case, we consider $r+\ell\leq 4c$. Algorithm~\ref{algo:pseudocode-gkp} performs a brute-force search to find minimum quotients of $f$ in Step~\ref{pseudocode-gkp-bruteforce} since $r\leq 4c$. Consequently, the collection returned by Algorithm~\ref{algo:pseudocode-gkp} contains $Q$ with probability $1$.

    We now prove the induction step where $r+\ell > 4c$. If $r\leq 4c$, Algorithm~\ref{algo:pseudocode-gkp} would terminate in Step~\ref{pseudocode-gkp-bruteforce} via brute-force search with no branching steps, contradicting $r+\ell > 4c$. Thus, $r>4c$. Moreover, we recall that $\cF$ is a $c$-\qbounded family of polymatroids and $f\in \cF$. This implies that
    \begin{align}
        |Q|\leq c\cdot \frac{\sum_{e\in \cN}f(e)}{r}. \label{inequality:gkp-quotient-bound-appendix}
    \end{align}
    
    We consider the following three cases:
    \begin{enumerate}
        \item $L_1=\emptyset$: In this case, all elements are light or moderate, i.e., $f(e)< \frac{r}{2c}$. Hence, we have
        \begin{align}
            \sum_{e \in \cN\setminus Q} f(e) &= \sum_{e \in \cN} f(e) - \sum_{e\in Q} f(e) \notag \\
            &\geq \sum_{e\in \cN} f(e) - \frac{r}{2c}\cdot |Q| \ \ \text{(since $f(e)<\frac{r}{2c}$)} \notag \\
            &\geq \frac{r}{c}\cdot |Q| - \frac{r}{2c}\cdot |Q| \ \ \text{(by inequality~(\ref{inequality:gkp-quotient-bound-appendix}))} \notag \\
            &= \frac{r}{2c}\cdot |Q|. \label{inequality:gkp-analysis-smallcase-appendix}
        \end{align}
        This further shows that
        $$\frac{r}{2c}\cdot |Q|\leq \sum_{e\in \cN\setminus Q}f(e)< \frac{r}{2c}\cdot |\cN\setminus Q|,$$
        which implies that $|Q|< |\cN\setminus Q|$ and $|Q|< |\cN|/2$. We observe that Algorithm~\ref{algo:pseudocode-gkp} will execute the if-clause in Step~\ref{pseudocode-gkp-if}. For each $e\not\in Q$, by Lemma~\ref{lemma:quotient-property} part 3, $Q$ is still a minimum non-empty quotient in the contracted polymatroid $f_e$. Moreover, the rank of the contracted polymatroid $f_e$ is $r-f(e) < r$ (since $f(e) > 0$) and hence by induction, $Q$ is in the collection returned by the algorithm executed on $f_e$ with probability at least $P(r-f(e),\ell)$.
        The algorithm outputs the set of quotients obtained from the recursive call on $f_e$.
        Therefore, 
        $$\begin{aligned}
            P(r, \ell) &\geq\frac{1}{|\cN|}\cdot \sum_{e\in \cN\setminus Q} P(r-f(e), \ell) \\
            &\geq \frac{1}{|\cN|} \cdot \sum_{e\in \cN\setminus Q} \frac{1}{(r-f(e))^{4c}} \cdot \left(\frac{1}{2(q+1)}\right)^{\ell} \ \ \text{(by induction hypothesis)} \\
            &\geq \frac{1}{|\cN|}\cdot \frac{|\cN\setminus Q|}{\left(r-\frac{\sum_{e\in \cN\setminus Q}f(e)}{|\cN\setminus Q|}\right)^{4c}}\cdot \left(\frac{1}{2(q+1)}\right)^{\ell} \ \ \text{(by convexity of the function $x^{-4c}$)} \\
            &\geq \frac{1}{|\cN|}\cdot \frac{|\cN\setminus Q|}{\left(r-\frac{r}{2c}\cdot \frac{|Q|}{|\cN\setminus Q|}\right)^{4c}}\cdot \left(\frac{1}{2(q+1)}\right)^{\ell} \ \ \text{(by inequality~(\ref{inequality:gkp-analysis-smallcase-appendix}))} \\
            &\geq \frac{1}{|\cN|}\cdot \frac{|\cN\setminus Q|}{\left(r-\frac{r}{2c}\cdot \frac{|Q|}{|\cN|}\right)^{4c}}\cdot \left(\frac{1}{2(q+1)}\right)^{\ell} \\
            &= \frac{1}{r^{4c}}\cdot \left(\frac{1}{2(q+1)}\right)^{\ell} \cdot \frac{1-\frac{|Q|}{|\cN|}}{\left(1-\frac{1}{2c}\cdot \frac{|Q|}{|\cN|}\right)^{4c}} \\
            &\geq \frac{1}{r^{4c}}\cdot \left(\frac{1}{2(q+1)}\right)^{\ell},
        \end{aligned}$$
        where the last inequality is because $\frac{|Q|}{|\cN|}\leq \frac{1}{2}$ and $(1-\frac{1}{2c}\cdot x)^{4c}\leq 1-x$ for every $x\in [0, 1/2)$.
        
        \item $L_1\neq \emptyset$ and $|L\setminus Q|=0$: In this case, Algorithm~\ref{algo:pseudocode-gkp} branches to $X=0$ with probability $\frac{1}{2}$. If $Q=L$, Algorithm~\ref{algo:pseudocode-gkp} would execute Step~\ref{pseudocode-gkp-singleton} since $Q=L$ is a non-empty quotient and hence $f(\cN \setminus L) < r$. 
        Thus, quotient $Q =L$ is contained in the collection returned by the algorithm. Suppose $Q\neq L$. This implies that $L\subsetneq Q$. Algorithm~\ref{algo:pseudocode-gkp} would recurse on polymatroid $f_{\backslash L}$ to obtain a collection of quotients of $f_{\backslash L}$ and add all elements of $L$ to every quotient in the collection. By Lemma~\ref{lemma:quotient-property} part 2, $Q\setminus L$ is a minimum non-empty quotient of $f_{\backslash L}$.
        By induction, this implies that the probability that quotient $Q\setminus L$ is in the collection returned by the algorithm when executed on polymatroid $f_{\backslash L}$ is at least $P(r, \ell-1)$. Moreover, if quotient $Q\setminus L$ is contained in the collection returned on input $f_{\backslash L}$, then quotient $Q$ would be contained in the collection returned by the algorithm on $f$. Therefore,
        $$\begin{aligned}
            P(r, \ell) &\geq \frac{1}{2}\cdot P(r, \ell -1) \\
            &\geq \frac{1}{2}\cdot \left(\frac{1}{r}\right)^{4c} \cdot \left(\frac{1}{2(q+1)}\right)^{\ell-1} \ \ \text{(by induction hypothesis)}\\
            &\geq \left(\frac{1}{r}\right)^{4c} \cdot \left(\frac{1}{2(q+1)}\right)^{\ell}.
        \end{aligned}$$

      \item $L_1\neq \emptyset$ and $|L\setminus Q|>0$: In this case, Algorithm~\ref{algo:pseudocode-gkp} branches to $X=1$ with probability $\frac{1}{2}$. Algorithm~\ref{algo:pseudocode-gkp} contracts an element $e\in L$, which is sampled uniformly at random. Thus, the probability that $e\not\in Q$ equals $\frac{|L\setminus Q|}{|L|}$. Since $e$ is heavy or moderate, $f(e)\geq \frac{r}{4c}$, which implies that contracting element $e$ reduces the rank by at least $\frac{r}{4c}$, and since $e \not \in Q$, by \ref{lemma:quotient-property}, $Q$ is in the collection returned by the algorithm executed on $f_e$ with probability at least $P(r-\frac{r}{4c},\ell)$. Further, if this happens, then $Q$ is in the output set of the algorithm.
        
        We claim that $\frac{|L\setminus Q|}{|L|}\geq \frac{1}{q+1}$. This is easy to see
        if $|L| \le q$ since $|L \setminus Q| \ge 1$. Otherwise, let $a = |L| - |Q| \ge 1$.
        Then $\frac{|L\setminus Q|}{|L|} \ge \frac{|L|-|Q|}{|L|} = \frac{a}{q+a} \ge 1/(q+1)$ for all $a \ge 1$. 

        Therefore,
        $$\begin{aligned}
            P(r, \ell) & \geq \frac{1}{2}\cdot \frac{|L\setminus Q|}{|L|} \cdot P\left(r-\frac{r}{4c}, \ell-1\right) \\
            &\geq \frac{1}{2}\cdot \frac{|L\setminus Q|}{|L|}\cdot \left(\frac{1}{\left(1-\frac{1}{4c}\right)r}\right)^{4c} \cdot \left(\frac{1}{2(q+1)}\right)^{\ell-1} \ \ \text{(by induction hypothesis)}\\
            &\geq \frac{1}{2}\cdot \frac{1}{q+1}\cdot \left(\frac{1}{\left(1-\frac{1}{4c}\right)r}\right)^{4c} \cdot \left(\frac{1}{2(q+1)}\right)^{\ell-1} \ \ \text{(since $\frac{|L\setminus Q|}{|L|}\geq \frac{1}{q+1}$)}\\
            &\geq \left(\frac{1}{r}\right)^{4c} \cdot \left(\frac{1}{2(q+1)}\right)^{\ell}.
        \end{aligned}$$
    \end{enumerate}
    Thus, inequality~(\ref{inequality:gkp-analysis-probability-appendix}) holds for every polymatroid $f\in \cF$. This implies that the collection of quotients returned by \text{\Search}$(f)$ contains $Q$ with probability at least
    $$P(r, \ell) \geq \left(\frac{1}{r}\right)^{4c} \cdot \left(\frac{1}{2(q+1)}\right)^{\ell}\geq \left(\frac{1}{r}\right)^{4c} \cdot \left(\frac{1}{2(q+1)}\right)^{4c\log r}=\left(\frac{1}{r}\right)^{O(c\log q)}.$$
\end{proof}

%% file: uniform-random-contraction-bounded-marginal.tex
\section{\MinQuo in $c$-\qbounded $k$-polymatroids}\label{section:uniform-contraction-bounded-marginal-appendix}

In this section, we consider $c$-\qbounded $k$-polymatroid families. We design a random contraction algorithm to find a subset $\cN'$ of the ground set such that the polymatroid obtained by contracting the complement of $\cN'$ has rank at most $c\alpha k$ and moreover, every $\alpha$-approximate minimum quotient $Q$ of $f$ is contained in $\cN'$ with probability $\Omega(k^{-1} r^{-c\alpha})$. The following is the main theorem of this section.

\begin{theorem}\label{thm:uniform-random-contraction}
    For an integer $k\ge 1$, let $\cF$ be a $c$-\qbounded $k$-polymatroid family and $\alpha \geq 1$. 
    There exists a randomized algorithm that takes a polymatroid $f\in \cF$ as input (via its evaluation oracle) and runs in polynomial time to return a subset $\cN'$ of the ground set $\cN$ of $f$ such that 
    \begin{enumerate}
        \item the polymatroid obtained from $f$ by contracting $\cN\setminus \cN'$ has rank at most $c\alpha k$ and 
        \item  every $\alpha$-approximate minimum quotient $Q$ of $f$ is contained in $\cN'$ with probability at least 
    $$\frac{c\alpha+1}{c\alpha k+1}\cdot \binom{r-c\alpha(k-1)}{c\alpha}^{-1}.$$
    \end{enumerate}
\end{theorem}

We first describe the random contraction algorithm. The algorithm here is a uniform random contraction algorithm contrasting with the non-uniform random contraction algorithms in previous sections. Let $f:2^{\cN}\rightarrow \mathbb{Z}_{\geq 0}$ be the given polymatroid from a $c$-\qbounded $k$-polymatroid family. Let $\alpha\geq 1$. If $f(\cN)\leq c\cdot \alpha\cdot k$, the algorithm returns the whole ground set $\cN$. Otherwise, among all elements $e\in \cN$ with $f(e)>0$, the algorithm samples one such element uniformly at random and recurses on the contracted polymatroid $f_e$. The pseudocode is given in Algorithm~\ref{algo:pseudocode-uniform-bounded-marginal}.

\begin{algorithm2e}[H]
\caption{Uniform Random Contraction Algorithm}
\label{algo:pseudocode-uniform-bounded-marginal}
\SetKwInput{KwInput}{Input}                
\SetKwInput{KwOutput}{Output}              
\LinesNumbered

\DontPrintSemicolon
  
    \KwInput{evaluation oracle access to polymatroid $f:2^{\cN}\rightarrow \Z_{\geq 0}$.
    }
    \KwOutput{a subset $S$ of the ground set $\cN$.}
    \SetKwFunction{Contract}{Contract}
    \SetKwProg{Fn}{Function}{:}{}
    
    \Fn{\Contract{$f$}}{

        \If{$f(\cN)\leq c\cdot \alpha\cdot k$}{
            \KwRet $\cN$.
        }

        $S\gets \{e\in \cN: f(e)>0\}$. \label{pseudocode-kk-delete}

        Sample element $e\in S$ uniformly at random.

        \KwRet $\Contract(f_e)$.
    }
\end{algorithm2e}

Let $\cN'\subseteq \cN$ be the subset returned by Algorithm~\ref{algo:pseudocode-uniform-bounded-marginal}. We observe that the polymatroid obtained from $f$ by contracting $\cN\setminus \cN'$ has rank at most $c\cdot \alpha\cdot k$.
Moreover, Algorithm~\ref{algo:pseudocode-uniform-bounded-marginal} takes polynomial time since at each recursive call to \Contract$(f)$, the size of the ground set decreases by at least $1$. We now analyze the probability that a fixed $\alpha$-approximate minimum quotient is contained in the returned subset. For positive integers $i,k,\alpha\geq 1$, we define
\[
    H(i,k,\alpha):= \begin{cases}
        \frac{c\alpha+1}{c\alpha k+1}\cdot \binom{i-c\alpha(k-1)}{c\alpha}^{-1}, & i>c\alpha k \\
        1, & otherwise.
    \end{cases}
\]
The following lemma will be used later. We defer the proof to Appendix~\ref{appendix:lp-analysis}.

\begin{restatable}{lemma}{lpanalysis}\label{lemma:lp-analysis-2}
    Let $r, k, c, \alpha\geq 1$ be positive integers with $r>c\cdot \alpha\cdot k$. Let $P:\mathbb{N}\rightarrow \mathbb{R}_{+}$ be a positive-valued function defined over the natural numbers. Then, the optimum value of the linear program~(\ref{lp-2}) defined below is at least $\min_{i=1}^{k} \frac{r-c\cdot \alpha\cdot k}{r+c\cdot \alpha\cdot (i-k)} \cdot P(r-i)$.
    \begin{equation}
    \label{lp-2}
    \begin{aligned}
        \min\quad &\sum_{i=1}^{k}(x_i-y_i)\cdot P(r-i) \\
        s.t.\quad & 0\leq y_i\leq x_i \quad \forall\ i \in [1, k] \\
            &\sum_{i=1}^{k}x_i=1\\
            &\sum_{i=1}^{k}y_i\leq \frac{c\cdot \alpha}{r}\cdot \sum_{i=1}^{k}i\cdot x_i
    \end{aligned}
    \end{equation}
\end{restatable}

Lemma~\ref{lemma:uniform-random-contraction} shown below implies Theorem~\ref{thm:uniform-random-contraction}.

\begin{lemma}\label{lemma:uniform-random-contraction}
    Let $\cF$ be a $c$-\qbounded $k$-polymatroid family with $\alpha,k \geq 1$. Let $f:2^{\cN}\rightarrow \mathbb{Z}_{\geq 0}$ be a polymatroid in $\cF$ with rank $r>c\alpha k$ and $q$ be the minimum size of non-empty quotients of $f$. Then, every $\alpha$-approximate minimum quotient $Q$ of $f$ is contained in the subset $\cN'$ returned by Algorithm~\ref{algo:pseudocode-uniform-bounded-marginal} on input $f$ with probability at least
    $$
    \frac{c\alpha+1}{c\alpha k+1}\cdot \binom{r-c\alpha(k-1)}{c\alpha}^{-1}.$$
\end{lemma}
\begin{proof}
    For every positive integer $r$, let $P(r)$ be the infimum over all $f\in \cF$ with rank $r$ and $f(e)\leq k$ for every element $e$ in the ground set of $f$ and all $\alpha$-approximate minimum quotient $Q$ of $f$, of the probability that $Q$ is contained in the subset $\cN'$ returned by $\Contract(f)$. We now prove that
    $$P(r)\geq H(r,k,\alpha)$$
    by induction on $r$. We recall that $f(e)\leq k$ for every $e\in \cN$.
    
    For the base case, we consider $r\leq c\alpha k$. Algorithm~\ref{algo:pseudocode-uniform-bounded-marginal} returns the whole ground set $\cN$, which implies that $P(r)=1\geq H(r,k,\alpha)$. Next, we prove the induction step. Let $q$ be the minimum size of non-empty quotients of $f$ and $|Q|\leq \alpha\cdot q$. Let $S$ be the set obtained in Line~\ref{pseudocode-kk-delete}, where $f(e)\geq 1$ for every $e\in S$. We observe that $Q\subseteq S$, since every element $e\in \cN$ with $f(e)=0$ cannot be in the quotient. For every integer $1\leq i \leq k$, we define
    $$x_i:=\frac{|\{e\in S: f(e)=i\}|}{|S|} \text{ \ and \ } y_i:=\frac{|\{e\in Q: f(e)=i\}|}{|S|},$$
    which further shows that $0\leq y_i\leq x_i$. Moreover, since $f\in \cF$ and $\cF$ is a $c$-\qbounded polymatroid family, we have 
    $$\begin{aligned}
        q \leq c\cdot \sum_{e\in \cN}\frac{f(e)}{r}=c\cdot \sum_{e\in S}\frac{f(e)}{r}.
    \end{aligned}$$
    This further shows that
    $$\begin{aligned}
        \sum_{i=1}^{k}y_i & = \frac{|Q|}{|S|} \leq \frac{\alpha\cdot q}{|S|} \leq \frac{c\cdot \alpha}{r} \cdot \sum_{e\in S}\frac{f(e)}{|S|} = \frac{c\cdot \alpha}{r} \cdot \sum_{i=1}^{k}i\cdot x_i.
    \end{aligned}$$
    We observe that Algorithm~\ref{algo:pseudocode-uniform-bounded-marginal} contracts an element $e\in S$, which is sampled uniform at random, and recurses on the contracted polymatroid $f_e$. By Lemma~\ref{lemma:quotient-property}, the minimum non-empty quotient of $f_e$ still has size at least $q$.
    Thus, if $e\not\in Q$, $Q$ is still a $\alpha$-approximate minimum quotient in $f_e$, and moreover, the rank of $f_e$ is $r-f(e)$ and hence, the set $Q$ is contained in the subset returned by the algorithm executed on $f_e$ with probability at least $P(r-f(e))$. The quotient $Q$ is contained in the subset returned by the algorithm if (i) the sampled element $e$ is not from $Q$ and (ii) $Q$ is contained in the subset returned by the algorithm when executed on the contracted polymatroid $f_e$. This implies that
    $$\begin{aligned}
        P(r)&\geq \sum_{i=1}^{k}\Pr[\ f(e)=i \text{ \ and \ } e\notin Q \ ] \cdot P(r-i) \\
        &= \sum_{i=1}^{k}(x_i-y_i)\cdot P(r-i).
    \end{aligned}$$
    Therefore, $P(r)$ is at least the optimum value of the following linear program:
    \begin{equation}
    \label{kogan-lp}
    \begin{aligned}
        \min\quad &\sum_{i=1}^{k}(x_i-y_i)\cdot P(r-i) \\
        s.t.\quad & 0\leq y_i\leq x_i \quad \forall\ i \in [1, k] \\
            &\sum_{i=1}^{k}x_i=1\\
            &\sum_{i=1}^{k}y_i\leq \frac{c\cdot \alpha}{r}\cdot \sum_{i=1}^{k}i\cdot x_i
    \end{aligned}
    \end{equation}

    By Lemma~\ref{lemma:lp-analysis-2}, there exists an index $1\leq i \leq k$ with
    $$P(r) \geq \frac{r-c\cdot \alpha\cdot k}{r+c\cdot \alpha\cdot (i-k)} \cdot P(r-i).$$

    If $r-i> c\cdot \alpha\cdot k$, then by defining $r':=r-c\cdot \alpha\cdot k$, we have
    $$\begin{aligned}
        \frac{H(r-i, k, \alpha)}{H(r, k, \alpha)} &= \frac{\binom{r'+c\alpha}{c\alpha}}{\binom{r'-i+c\alpha}{c\alpha}} = \frac{(r'+c\alpha)\ldots(r'-i+c\alpha+1)}{r'\ldots (r'-i+1)} \\
        &\geq \left(1+\frac{c\alpha}{r'}\right)^i \geq 1+\frac{c\alpha}{r'}\cdot i = \frac{r+c\cdot \alpha\cdot (i-k)}{r-c\cdot \alpha\cdot k},
    \end{aligned}$$
    which further shows that
    $$\begin{aligned}
        P(r) &\geq \frac{r-c\cdot \alpha\cdot k}{r+c\cdot \alpha\cdot (i-k)} \cdot P(r-i) \\
        &\geq \frac{r-c\cdot \alpha\cdot k}{r+c\cdot \alpha\cdot (i-k)} \cdot H(r-i, k, \alpha)\ \ \text{(by induction hypothesis)}\\
        &\geq H(r, k, \alpha).
    \end{aligned}$$

    If $r-i\leq c\cdot \alpha\cdot k$, then $P(r-i)=1$ and we have
    $$\begin{aligned}
        P(r) &\geq \frac{r-c\cdot \alpha\cdot k}{r+c\cdot \alpha\cdot (i-k)} \cdot P(r-i) = \frac{r-c\cdot \alpha\cdot k}{r+c\cdot \alpha\cdot (i-k)}\\
        &= 1 - \frac{c\cdot \alpha\cdot i}{r+c\cdot \alpha\cdot (i-k)} \geq 1 - \frac{c\cdot \alpha\cdot i}{1+c\cdot \alpha\cdot i} \ \ \text{(since $r\geq c\cdot \alpha\cdot k +1$)} \\
        &\geq \frac{1}{c\cdot \alpha\cdot k +1} \geq \frac{1}{c\cdot \alpha\cdot k +1}\cdot \frac{c\alpha+1}{\binom{r-c\alpha(k-1)}{c\alpha}} = H(r, k, \alpha),
    \end{aligned}$$
    where the last inequality is because $r\geq c\alpha k+1$ and $\binom{r-c\alpha(k-1)}{c\alpha}\geq \binom{c\alpha+1}{c\alpha}=c\alpha+1$.

\end{proof}

%% file: lp-analysis.tex
\section{Proof of Lemmas~\ref{lemma:lp-analysis-1} and \ref{lemma:lp-analysis-2}}\label{appendix:lp-analysis}
In this section, we prove Lemmas~\ref{lemma:lp-analysis-1} and \ref{lemma:lp-analysis-2}.

\lpanalysisstrong*
\begin{proof}
    We observe that the last constraint implies that
    \begin{align}
        \sum_{i=1}^{k}y_i\leq \frac{\alpha}{r}\cdot \sum_{i=1}^{k}(i+c)\cdot x_i\leq \frac{\alpha \cdot (c+k)}{r}\cdot \sum_{i=1}^{k}x_i < \sum_{i=1}^{k}x_i, \label{inequality:kogan-lp-constraint-strong}
    \end{align}
    where the last inequality is because $r>\alpha\cdot (c+k)$. This shows that in every feasible solution, there is always some $y_i<x_i$. Hence, in every optimal solution, the last constraint is tight (otherwise, increasing $y_i$ would decrease the value of the solution without violating any other constraint). We recall that we have $2k$ variables and $2k+2$ constraints in LP-(\ref{lp-1}). This implies that in every optimal solution, there are at most $2$ non-tight constraints among the $2k$ constraints $0\leq y_i\leq x_i$. We proceed by analyzing the possible cases:
    \begin{enumerate}
        \item $0<y_i=x_i$ and $0<y_j=x_j$: In this case, we have $\sum_{i=1}^{k}x_i=\sum_{i=1}^{k}y_i$, which contradicts to inequality~(\ref{inequality:kogan-lp-constraint-strong}).
        
        \item $0=y_i<x_i$ and $0=y_j<x_j$: In this case, we have $\sum_{i=1}^{k}y_i=0$, which is also impossible since we may increase $y_i$ to decrease the value of the solution without violating any constraint.

        \item $0=y_i<x_i$ and $0<y_j=x_j$: In this case, we have $x_i+x_j=1$ and $x_j=y_j=\frac{\alpha}{r}\cdot ((i+c)\cdot x_i+(j+c)\cdot x_j)$, which implies that
        $$x_i=\frac{1-\frac{\alpha}{r}\cdot (j+c)}{1+\frac{\alpha}{r}\cdot (i-j)} \text{ \ and \ } x_j=\frac{\frac{\alpha}{r}\cdot (i+c)}{1+\frac{\alpha}{r}\cdot (i-j)}.$$
        This shows that
        \begin{align}
            \mathsf{LP} &= x_i\cdot P(r-i) = \left(1- \frac{\alpha\cdot (i+c)}{r+\alpha\cdot (i-j)}\right) \cdot P(r-i) \notag\\
            &\geq \left(1- \frac{\alpha\cdot (i+c)}{r+\alpha\cdot (i-k)}\right) \cdot P(r-i) \notag\\
            &= \frac{r-\alpha\cdot (c+k)}{r+\alpha\cdot (i-k)} \cdot P(r-i). \label{inequality:kogan-lp-analysis-strong}
        \end{align}

        \item $0<y_i<x_i$: In this case, we have $x_i=1$ and $y_i=\frac{\alpha\cdot (i+c)}{r}$. Thus, we have
        $$\begin{aligned}
            \mathsf{LP} &= \left(1-\frac{\alpha\cdot (i+c)}{r}\right) \cdot P(r-i) \geq \left(1-\frac{\alpha\cdot (i+c)}{r+\alpha\cdot (i-k)}\right)\cdot P(r-i) \\
            &= \frac{r-\alpha\cdot (c+k)}{r+\alpha\cdot (i-k)} \cdot P(r-i).
        \end{aligned}$$
    \end{enumerate}
\end{proof}

\lpanalysis*
\begin{proof}
    We observe that the last constraint implies that
    \begin{align}
        \sum_{i=1}^{k}y_i\leq \frac{c\cdot \alpha}{r}\cdot \sum_{i=1}^{k}i\cdot x_i\leq \frac{c\cdot \alpha \cdot k}{r}\cdot \sum_{i=1}^{k}x_i < \sum_{i=1}^{k}x_i, \label{inequality:kogan-lp-constraint}
    \end{align}
    where the last inequality is because $r>c\cdot \alpha\cdot k$. This shows that in every feasible solution, there is always some $y_i<x_i$. Hence, in every optimal solution, the last constraint is tight (otherwise, increasing $y_i$ would decrease the value of the solution without violating any other constraint). We recall that we have $2k$ variables and $2k+2$ constraints in LP-(\ref{lp-2}). This implies that in every optimal solution, there are at most $2$ non-tight constraints among the $2k$ constraints $0\leq y_i\leq x_i$. We proceed by analyzing the possible cases:
    \begin{enumerate}
        \item $0<y_i=x_i$ and $0<y_j=x_j$: In this case, we have $\sum_{i=1}^{k}x_i=\sum_{i=1}^{k}y_i$, which contradicts to inequality~(\ref{inequality:kogan-lp-constraint}).
        
        \item $0=y_i<x_i$ and $0=y_j<x_j$: In this case, we have $\sum_{i=1}^{k}y_i=0$, which is also impossible since we may increase $y_i$ to decrease the value of the solution without violating any constraint.

        \item $0=y_i<x_i$ and $0<y_j=x_j$: In this case, we have $x_i+x_j=1$ and $x_j=y_j=\frac{c\cdot \alpha}{r}\cdot (i\cdot x_i+j\cdot x_j)$, which implies that
        $$x_i=\frac{1-\frac{c\cdot \alpha}{r}\cdot j}{1+\frac{c\cdot \alpha}{r}\cdot (i-j)} \text{ \ and \ } x_j=\frac{\frac{c\cdot \alpha}{r}\cdot i}{1+\frac{c\cdot \alpha}{r}\cdot (i-j)}.$$
        This shows that
        \begin{align}
            \mathsf{LP} &= x_i\cdot P(r-i) = \left(1- \frac{c\cdot \alpha\cdot i}{r+c\cdot \alpha\cdot (i-j)}\right) \cdot P(r-i) \notag\\
            &\geq \left(1- \frac{c\cdot \alpha\cdot i}{r+c\cdot \alpha\cdot (i-k)}\right) \cdot P(r-i) \notag\\
            &= \frac{r-c\cdot \alpha\cdot k}{r+c\cdot \alpha\cdot (i-k)} \cdot P(r-i). \label{inequality:kogan-lp-analysis}
        \end{align}

        \item $0<y_i<x_i$: In this case, we have $x_i=1$ and $y_i=\frac{c\cdot\alpha\cdot i}{r}$. Thus, we have
        $$\begin{aligned}
            \mathsf{LP} &= \left(1-\frac{c\cdot \alpha\cdot i}{r}\right) \cdot P(r-i) \geq \left(1-\frac{c\cdot\alpha\cdot i}{r+c\cdot \alpha\cdot (i-k)}\right)\cdot P(r-i) \\
            &= \frac{r-c\cdot \alpha\cdot k}{r+c\cdot \alpha\cdot (i-k)} \cdot P(r-i).
        \end{aligned}$$
    \end{enumerate}    
\end{proof}

%% file: references.bib
@inproceedings{KW21,
author = {David R. Karger and David P. Williamson},
title = {Recursive Random Contraction Revisited},
booktitle = {Symposium on Simplicity in Algorithms (SOSA)},
pages = {68-73},
year = {2021},
}

@article{Kar99-sampling,
author = {Karger, D. R.},
year = {1999},
title = {{Random Sampling in Cut, Flow, and Network Design Problems}},
journal = {Mathematics of Operations Research},
volume = {24},
number = {2},
pages = {383-413},
}

@inproceedings{KKTY21-STOC,
  title={Towards tight bounds for spectral sparsification of hypergraphs},
  author={Kapralov, M. and Krauthgamer, R. and Tardos, J. and Yoshida, Y.},
  booktitle={Proceedings of the 53rd Annual ACM SIGACT Symposium on Theory of Computing},
  series = {STOC},
  pages={598-611},
  year={2021}
}

@inproceedings{KKTY21-FOCS, 
  title={{Spectral Hypergraph Sparsifiers of Nearly Linear Size}},
  author={Kapralov, M. and Krauthgamer, R. and Tardos, J. and Yoshida, Y.},
  booktitle={IEEE 62nd Annual Symposium on Foundations of Computer Science},
  series = {FOCS},
  pages={1159-1170},
  year={2022},
}

@inproceedings{BST19,
  title={New notions and constructions of sparsification for graphs and hypergraphs},
  author={Bansal, N. and Svensson, O. and Trevisan, L.},
  booktitle={IEEE 60th Annual Symposium on Foundations of Computer Science},
  series = {FOCS},
  pages={910-928},
  year={2019},
}

@inproceedings{JLS23,
author = {Jambulapati, A. and Liu, Y. and Sidford, A.},
title = {{Chaining, Group Leverage Score Overestimates, and Fast Spectral Hypergraph Sparsification}},
year = {2023},
booktitle = {Proceedings of the 55th Annual ACM Symposium on Theory of Computing},
pages = {196-206},
series = {STOC}, 
}

@inproceedings{Lee23,
author = {Lee, James R.},
title = {Spectral Hypergraph Sparsification via Chaining},
year = {2023},
booktitle = {Proceedings of the 55th Annual ACM Symposium on Theory of Computing},
pages = {207-218},
series = {STOC}, 
}

@INPROCEEDINGS{KKTY22,
  author={Kapralov, M. and Krauthgamer, R. and Tardos, J. and Yoshida, Y.},
  booktitle={Proceedings of the 62nd Annual IEEE Symposium on Foundations of Computer Science}, 
  series = {FOCS},
  title={Spectral Hypergraph Sparsifiers of Nearly Linear Size}, 
  year={2022},
  pages={1159-1170},
}

@inproceedings{CLP24,
author = {Cen, Ruoxu and Li, Jason and Panigrahi, Debmalya},
title = {Hypergraph Unreliability in Quasi-Polynomial Time},
year = {2024},
booktitle = {Proceedings of the 56th Annual ACM Symposium on Theory of Computing},
pages = {1700–1711},
}

@article{CQX20,
author = {Chekuri, C. and Quanrud, K. and Xu, C.},
title = {{LP} Relaxation and Tree Packing for Minimum $k$-Cut},
journal = {SIAM Journal on Discrete Mathematics},
volume = {34},
number = {2},
pages = {1334-1353},
year = {2020},
}

@article{NNI97,
 author = {Nagamochi, H. and Nishimura, K. and Ibaraki, T.}, 
 title = {Computing all small cuts in an undirected network}, 
 journal = {SIAM Journal on Discrete Mathematics},
 volume = {10},
 year = {1997}, 
 number = {3}, 
 pages = {469–481},
}

@inproceedings{BCW23,
author = {Beideman, Calvin and Chandrasekaran, Karthekeyan and  Wang, Weihang},
title = {Approximate minimum cuts and their enumeration},
booktitle = {Symposium on Simplicity in Algorithms (SOSA)},
pages = {36-41},
year = {2023},
}

@INPROCEEDINGS{Kar17,
  author={Karger, David R.},
  booktitle={IEEE 58th Annual Symposium on Foundations of Computer Science (FOCS)}, 
  title={Faster (and Still Pretty Simple) Unbiased Estimators for Network (Un)reliability}, 
  year={2017},
  pages={755-766},
}

@inproceedings{Kar20,
author = {Karger, David R.},
title = {A phase transition and a quadratic time unbiased estimator for network reliability},
year = {2020},
booktitle = {Proceedings of the 52nd Annual ACM SIGACT Symposium on Theory of Computing},
pages = {485–495},
}

@inproceedings{CLP25,
author = {Cen, Ruoxu and Li, Jason and Panigrahi, Debmalya},
title = {Network Unreliability in Almost-Linear Time},
year = {2025},
booktitle = {Proceedings of the 57th Annual ACM Symposium on Theory of Computing},
pages = {120–131},
}

@article{Kar01,
author = {Karger, David R.},
title = {A Randomized Fully Polynomial Time Approximation Scheme for the All-Terminal Network Reliability Problem},
journal = {SIAM Journal on Computing},
volume = {43},
number = {3},
pages = {499-522},
year = {2001},
}

@article{GHLL21,
author = {Gupta, Anupam and Harris, David G. and Lee, Euiwoong and Li, Jason},
title = {Optimal Bounds for the k-cut Problem},
year = {2021},
volume = {69},
number = {1},
journal = {J. ACM},
articleno = {2},
}

@article{BCX23,
author = {Beideman, C. and Chandrasekaran, K. and Xu, C.},
title = {Multicriteria cuts and size-constrained k-cuts in hypergraphs},
journal ={Math. Program.},
volume = {197}, 
pages = {27–69},
year = {2023},
}

@inproceedings{Kar16,
author = {Karger, David R.},
title = {Enumerating parametric global minimum cuts by random interleaving},
year = {2016},
booktitle = {Proceedings of the Forty-Eighth Annual ACM Symposium on Theory of Computing},
pages = {542–555},
}

@inproceedings{AMR17,
  author       = {Aissi, H. and
                  Mahjoub, A. and
                  R. Ravi},
  title        = {Randomized Contractions for Multiobjective Minimum Cuts},
  booktitle    = {25th Annual European Symposium on Algorithms, {ESA}},
  pages        = {6:1--6:13},
  year         = {2017},
}

@article{Var97,
  author={Vardy, A.},
  journal={IEEE Transactions on Information Theory}, 
  title={The intractability of computing the minimum distance of a code}, 
  year={1997},
  volume={43},
  number={6},
  pages={1757-1766},
}

@article{ZCTZ11,
author = {Zhang, P. and Cai, J.-Y. and Tang, L.-Q. and Zhao, W.-B},
title = {Approximation and hardness results for label cut and
related problems},
year = {2011},
volume = {21},
number = {2},
journal = {Journal of Combinatorial Optimization},
pages = {192--208},
}

@article{Kar-matroid,
 title = {Random sampling and greedy sparsification for matroid optimization problems},
 author = {Karger, D.},
journal = {Mathematical Programming},
volume = {82},
pages = {41--81},
year = {1998}
}

@article{JLMPTS26,
author = {Jaffke, Lars and T. de Lima, Paloma and Masa\v{r}\'{\i}k, Tom\'{a}\v{s} and Pilipczuk, Marcin and Souza, Uev\'{e}rton S.},
title = {A Tight Quasi-Polynomial Bound for Global Label Min-Cut},
year = {2026},
volume = {22},
number = {2},
journal = {ACM Trans. Algorithms},
articleno = {23},
}

@article{BK15,
author = {Bencz\'{u}r, Andr\'{a}s A. and Karger, David R.},
title = {Randomized Approximation Schemes for Cuts and Flows in Capacitated Graphs},
journal = {SIAM Journal on Computing},
volume = {44},
number = {2},
pages = {290-319},
year = {2015},
doi = {10.1137/070705970},
}

@phdthesis{mccormick-thesis,
    author = {McCormick, T.},
    title = {{A Combinatorial Approach to Some Sparse Matrix Problems}},
    school = {{Stanford University}},
    year = {1983},
}

@phdthesis{karger1995random,
  title={Random sampling in graph optimization problems},
  author={Karger, D.},
  year={1995},
  school={Stanford University}
}

@inproceedings{karger1993global,
  title={{Global Min-cuts in RNC, and Other Ramifications of a Simple Min-Cut Algorithm}},
  author={Karger, David R},
  booktitle={Proceedings of the fourth annual ACM-SIAM symposium on Discrete algorithms},
  series = {SODA},
  pages={21--30},
  year={1993}
}

@article{KS96,
  title={A new approach to the minimum cut problem},
  author={Karger, David R and Stein, Clifford},
  journal={Journal of the ACM (JACM)},
  volume={43},
  number={4},
  pages={601--640},
  year={1996},
  publisher={ACM New York, NY, USA}
}

@inproceedings{GKP17,
  title={Random contractions and sampling for Hypergraph and Hedge Connectivity},
  author={Ghaffari, M. and Karger, D. and Panigrahi, D.},
  booktitle={Proceedings of the Twenty-Eighth Annual ACM-SIAM Symposium on Discrete Algorithms},
  pages={1101--1114},
  year={2017}
}

@article{CXY21,
	author = {Chandrasekaran, Karthekeyan and Xu, Chao and Yu, Xilin},
	title = {{Hypergraph $k$-Cut in randomized polynomial time}},
	journal = {Mathematical Programming},
	volume = {186}, 
	pages = {85-113}, 
	year = {2021},
}

@article{fox2023minimum,
  title={Minimum cut and minimum k-cut in hypergraphs via branching contractions},
  author={Fox, Kyle and Panigrahi, Debmalya and Zhang, Fred},
  journal={ACM Transactions on Algorithms},
  volume={19},
  number={2},
  pages={1--22},
  year={2023},
  publisher={ACM New York, NY}
}

@misc{chandrasekaran2025hedgegraph,
  title={Hedgegraph Polymatroids},
  author={Chandrasekaran, Karthekeyan and Chekuri, Chandra and Wang, Weihang and Zhu, Weihao},
  howpublished={arXiv preprint arXiv:2510.25043},
  year={2025}
}

@inproceedings{fomin2025fixed,
  title={Fixed-parameter tractability of hedge cut},
  author={Fomin, Fedor V and Golovach, Petr A and Korhonen, Tuukka and Lokshtanov, Daniel and Saurabh, Saket},
  booktitle={Proceedings of the 2025 Annual ACM-SIAM Symposium on Discrete Algorithms (SODA)},
  pages={1402--1411},
  year={2025},
  organization={SIAM}
}

@inproceedings{quanrud2024quotient,
  title={Quotient sparsification for submodular functions},
  author={Quanrud, Kent},
  booktitle={Proceedings of the 2024 Annual ACM-SIAM Symposium on Discrete Algorithms (SODA)},
  pages={5209--5248},
  year={2024},
  organization={SIAM}
}

@inproceedings{kogan2015sketching,
  title={Sketching cuts in graphs and hypergraphs},
  author={Kogan, Dmitry and Krauthgamer, Robert},
  booktitle={Proceedings of the 2015 Conference on Innovations in Theoretical Computer Science},
  pages={367--376},
  year={2015}
}

@inproceedings{chen2020near,
  title={Near-linear size hypergraph cut sparsifiers},
  author={Chen, Yu and Khanna, Sanjeev and Nagda, Ansh},
  booktitle={2020 IEEE 61st Annual Symposium on Foundations of Computer Science (FOCS)},
  pages={61--72},
  year={2020},
  organization={IEEE}
}

@book{Schrijver-book,
  title={Combinatorial optimization: polyhedra and efficiency},
  author={Schrijver, Alexander and others},
  year={2003},
  publisher={Springer}
}

@article{Zaslavsky,
  author  = {Zaslavsky, Thomas},
  title   = {Biased graphs. {II}. {T}he three matroids},
  journal = {Journal of Combinatorial Theory, Series B},
  volume  = {51},
  number  = {1},
  pages   = {46--72},
  year    = {1991},
  doi     = {10.1016/0095-8956(91)90005-5}
}

@article{GGW18,
  author  = {Geelen, Jim and Gerards, Bert and Whittle, Geoff},
  title   = {Quasi-graphic matroids},
  journal = {Journal of Graph Theory},
  volume  = {87},
  number  = {2},
  pages   = {253--264},
  year    = {2018},
  doi     = {10.1002/jgt.22177},
  note    = {The originally published version (2017) was retracted due to an error in Lemma~3.5; the corrected version appears in volume~87.}
}

@article{GeelenKapadia,
  author  = {Geelen, Jim and Kapadia, Rohan},
  title   = {Computing girth and cogirth in perturbed graphic matroids},
  journal = {Combinatorica},
  volume  = {38},
  number  = {1},
  pages   = {167--191},
  year    = {2018},
  doi     = {10.1007/s00493-016-3445-3},
  eprint  = {1504.07647},
  archivePrefix = {arXiv},
  primaryClass  = {math.CO}
}

@article{spielman2011spectral,
  title={Spectral sparsification of graphs},
  author={Spielman, Daniel A and Teng, Shang-Hua},
  journal={SIAM Journal on Computing},
  volume={40},
  number={4},
  pages={981--1025},
  year={2011},
  publisher={SIAM}
}

@article{spielman2011graph,
  title={Graph sparsification by effective resistances},
  author={Spielman, Daniel A and Srivastava, Nikhil},
  journal={SIAM Journal on Computing},
  volume={40},
  number={6},
  pages={1913--1926},
  year={2011},
  publisher={SIAM}
}

@article{batson2012twice,
  title={Twice-ramanujan sparsifiers},
  author={Batson, Joshua and Spielman, Daniel A and Srivastava, Nikhil},
  journal={SIAM Journal on Computing},
  volume={41},
  number={6},
  pages={1704--1721},
  year={2012},
  publisher={SIAM}
}

@article{barahona1987construction,
  title={A construction for binary matroids},
  author={Barahona, Francisco and Conforti, Michele},
  journal={Discrete Mathematics},
  volume={66},
  number={3},
  pages={213--218},
  year={1987},
  publisher={Elsevier}
}

@article{geelen2015highly,
  title={The highly connected matroids in minor-closed classes},
  author={Geelen, Jim and Gerards, Bert and Whittle, Geoff},
  journal={Annals of Combinatorics},
  volume={19},
  number={1},
  pages={107--123},
  year={2015},
  publisher={Springer}
}
